\documentclass[journal,twoside,web]{ieeecolor}
\usepackage{generic}
\usepackage{cite}
\usepackage{amsmath,amssymb,amsfonts}
\usepackage{algorithmic}
\usepackage{graphicx}
\usepackage{textcomp}

\usepackage{epsfig}
\usepackage{cleveref}
\usepackage[caption=false]{subfig}
\usepackage{xcolor}
\DeclareMathOperator*{\argmax}{arg\,max}
\DeclareMathOperator*{\argmin}{arg\,min}
\newtheorem{theorem}{Theorem}
\newtheorem{lemma}{Lemma}
\newtheorem{corollary}{Corollary}
\newtheorem{remark}{Remark}
\newtheorem{definition}{Definition}

\usepackage{mathtools}

\begin{document}
\title{Ensuring Resilience Against Stealthy Attacks on Cyber-Physical Systems}
\author{Paul Griffioen, \IEEEmembership{Student Member, IEEE}, Bruce H. Krogh, \IEEEmembership{Fellow, IEEE}, and Bruno Sinopoli, \IEEEmembership{Fellow, IEEE}
\thanks{Copyright 2022 IEEE.
This material is based upon work funded and supported by the Department of Defense under Contract No. FA8702-15-D-0002 with Carnegie Mellon University for the operation of the Software Engineering Institute, a federally funded research and development center.
DM22-0372
}
\thanks{P. Griffioen and B. H. Krogh are with the Department of Electrical and Computer Engineering, Carnegie Mellon University, Pittsburgh, PA, USA 15213. B. Sinopoli is with the Department of Electrical and Systems Engineering, Washington University in St. Louis, St. Louis, MO, USA 63130. Email: {\tt\small\{pgriffi1|krogh\}@andrew.cmu.edu, bsinopoli@wustl.edu}}}

\maketitle

%%%%%%%%%%%%%%%%%%%%%%%%%%%%%%%%%%%%%%%%%%%%%%%%%%%%%%%%%
\begin{abstract}
This article provides a tool for analyzing mechanisms that aim to achieve resilience against stealthy, or undetectable, attacks on cyber-physical systems (CPSs). We consider attackers who are able to corrupt all of the inputs and outputs of the system. To counter such attackers, a response scheme must be implemented that keeps the attacker from corrupting the inputs and outputs of the system for certain periods of time. To aid in the design of such a response scheme, our tool provides sufficient lengths for these periods of time in order to ensure safety with a particular probability. We provide a conservative upper bound on how long the system can remain under stealthy attack before the safety constraints are violated. Furthermore, we show how a detector limits the set of biases an attacker can exert on the system while still remaining stealthy, aiding a system operator in the design of the detector. Our contributions are demonstrated with an illustrative example.
\end{abstract}

%%%%%%%%%%%%%%%%%%%%%%%%%%%%%%%%%%%%%%%%%%%%%%%%%%%%%%%%%

%%%%%%%%%%%%%%%%%%%%%%%%%%%%%%%%%%%%%%%%%%%%%%%%%%%%%%%%%
\section{Introduction}
Securing cyber-physical systems, engineered systems which include sensing, processing, control, and communication in physical spaces, is essential in today’s society. CPSs are ubiquitous in modern critical infrastructures including manufacturing, transportation systems, energy delivery, health care, water management, and the smart grid. The presence of heterogeneous components and devices creates numerous attack surfaces in these large scale, highly connected systems. Consequently, these systems are attractive targets for adversaries and are essential to protect.

Due to the strong coupling between cyber and physical domains, the tools and methodologies developed to ensure security in the cyber domain alone are insufficient to secure CPSs. Techniques within cyber security such as authenticated encryption, message authentication codes, and digital signatures are often computationally expensive to implement and fail to recognize purely physical attacks. For example, the integrity of sensor measurements can be modified by changing a sensor’s local environment while control inputs can be changed by directly manipulating system actuators. In such a scenario, message authentication codes or digital signatures fail to recognize an attack.

CPS vulnerabilities have culminated in several effective attacks from highly resourceful and knowledgeable adversaries. In the Maroochy Shire incident \cite{slay2007lessons}, a malicious insider was able to utilize detailed system knowledge to attack a waste management system in Queensland, Australia, resulting in the leakage of millions of liters of sewage. With the Stuxnet attack \cite{langner2011stuxnet}, a nation state adversary was able to compromise a uranium enrichment facility in Iran, leading to the destruction of a thousand centrifuges. In 2015, hackers were able to remotely compromise a supervisory control and data acquisition (SCADA) system in Ukraine \cite{case2016analysis}, allowing them to cause widespread blackouts.

Motivated by the threat of such sophisticated attackers, we aim to design resilient CPSs by providing a set of mechanisms and tools that can be used to preserve safety while functionality is restored in the presence of attacks. In this work, we present a tool that provides a conservative upper bound on the amount of time a system can be under stealthy attack while still remaining probabilistically safe. This tool is general enough that it can be used in the design and analysis of any resilience mechanisms.

The rest of this article is organized as follows. Section II surveys a variety of mechanisms used to achieve resilience, specifically those used for detection and response. Section III introduces the system model used in analyzing these response mechanisms, along with the control, estimation, and detection schemes. Section IV describes stealthy adversaries, analyzing the set of biases they are able to exert on the system without being detected. Section V provides sufficient conditions for ensuring the probabilistic safety of the overall system, describing how the length of time under which the components can be trusted directly affects whether or not these conditions are satisfied. Section VI includes results from an example illustration, and Section VII concludes the article.

\section{Previous Work}
Two necessary components in designing resilient CPSs include detection and response. The recognition and detection of attacks is the first and foremost step in achieving resilience. Once an attack is detected, a number of forms of active response can be implemented to ensure system resilience.

\subsection{Detection}
There are two main forms of detection: passive and active. Passive detection techniques process the defender’s information to make a binary decision about the system, outputting either the null hypothesis (normal system operation) or the alternative hypothesis (the system is under attack). While passive detection techniques are effective against benign faults, bad data detectors can oftentimes be bypassed by powerful adversaries who leverage access to system channels and/or model knowledge to construct attacks so that the outputs received by a system operator are statistically consistent with expected output behavior. Since passive detection techniques provably fail in these instances, active detection techniques must be used to detect malicious attacks, where a defender intelligently changes the policy online by adding perturbations to the system \cite{griffioen2019tutorial,weerakkody2017active}.

Figure \ref{fig:ActiveDetection} categorizes two types of active detection according to which part of the control infrastructure is being attacked.
\begin{figure}[h!]
\centering
\subfloat[Physical Watermarking]{\includegraphics[width=0.5\columnwidth]{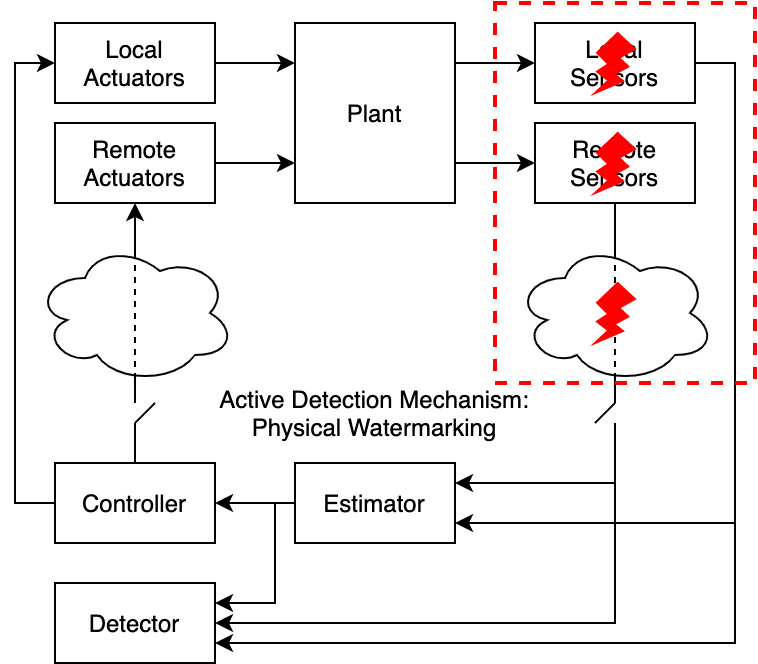}\label{fig:Watermarking}}
\subfloat[Moving Target Defense]{\includegraphics[width=0.5\columnwidth]{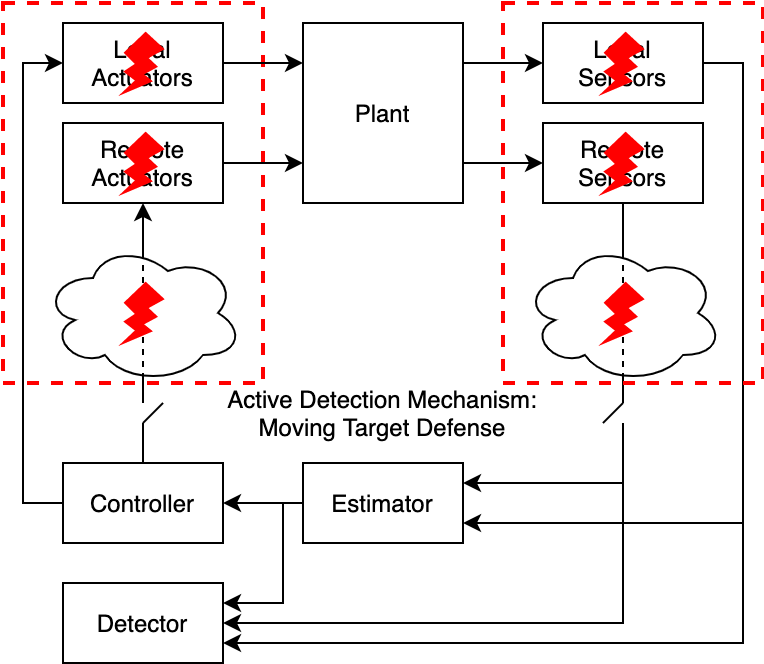}\label{fig:MovingTarget}}
\caption{Active detection mechanism for attacks on (a) the sensor measurements and (b) the control inputs and sensor measurements.}
\label{fig:ActiveDetection}
\end{figure}
When the physical sensors and/or the sensor communication channels are being attacked (Figure \ref{fig:Watermarking}), physical watermarking can be used for detection. When the physical actuators and sensors and/or the actuator and sensor communication channels are being attacked (Figure \ref{fig:MovingTarget}), the moving target defense can be used for detection. In this way, the moving target defense is effective in detecting a broader range of attacks.

Motivated by the use of nonces in cyber security, a physical watermark is a secret noisy (random) control input inserted in addition to or in place of an intended control input to authenticate the system \cite{mo2009secure,mo2013detecting,mo2015physical,satchidanandan2016dynamic,weerakkody2014detecting,ferrari2017detection,fang2020optimal}. In particular, the control input serves as a secret and the watermark acts as a cyber-physical nonce. Under normal conditions, the watermark will be embedded in the sensor outputs due to the system dynamics. However, under replay attack, the measurements will contain physical responses to an earlier sequence of watermarks. By designing a detector that recognizes the presence of the watermark in the sensor outputs, a defender is able to verify the freshness of the received sensor measurements.

Motivated by the use of message authentication codes (MACs) in cyber security, the moving target defense keeps an adversary unaware of the full system model by changing the parameters of the system’s physical dynamics (hybrid moving target defense) or by changing the parameters of an authenticating subsystem’s physical dynamics (extended/nonlinear moving target defense) \cite{griffioen2020moving,kanellopoulos2019moving,tian2019moving,weerakkody2015detecting,weerakkody2016moving,griffioen2019optimal,schellenberger2017detection,giraldo2019moving}. The time-varying sequence of parameters is kept hidden from the adversary, functioning as a moving target so that the adversary is kept from identifying the physical dynamics. This forces an attacker to leverage imperfect system information when constructing an attack, which in turn can reveal the attacker’s malicious behavior.

\subsection{Response}
While the recognition and detection of attacks is crucial, it is not sufficient for guaranteeing the safety and security of the CPS when under attack. Consequently, it is necessary to develop mechanisms and strategies that can be implemented to ensure resilience against attacks. Figure \ref{fig:Response} categorizes a variety of response mechanisms according to which part of the control infrastructure is being attacked.
\begin{figure}[h!]
\centering
\subfloat[Software Rejuvenation]{\includegraphics[width=0.33\columnwidth]{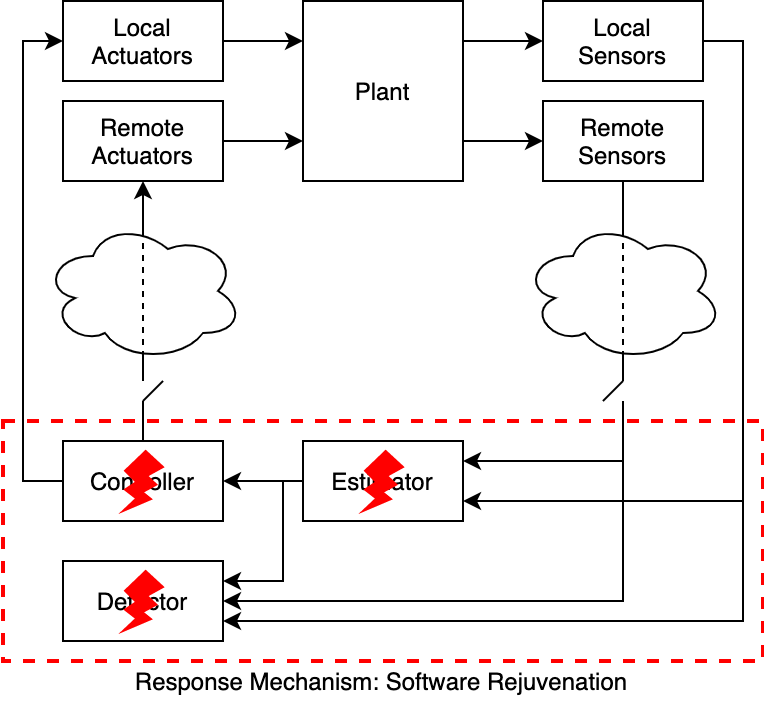}\label{fig:SoftwareRejuvenation}}
\subfloat[Overlay Networks]{\includegraphics[width=0.33\columnwidth]{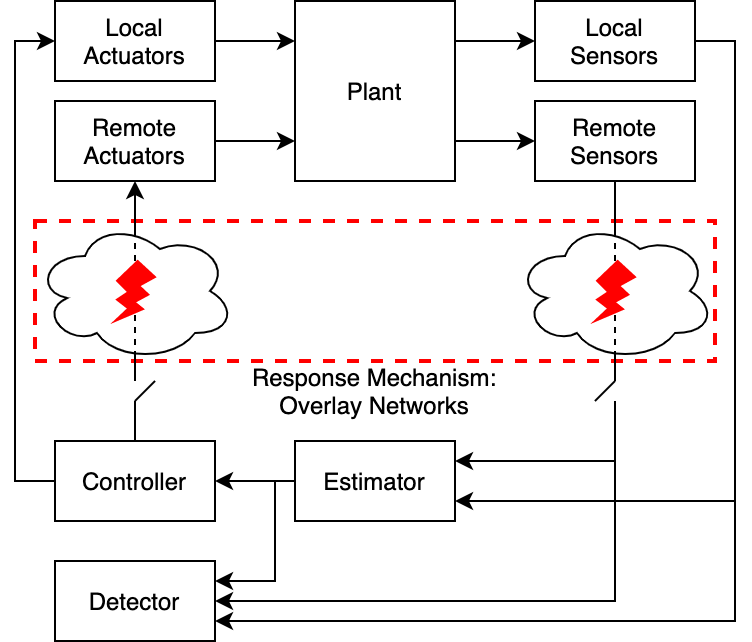}\label{fig:OverlayNetworks}}
\subfloat[Reconfigurable Control and Resilient Estimation]{\includegraphics[width=0.33\columnwidth]{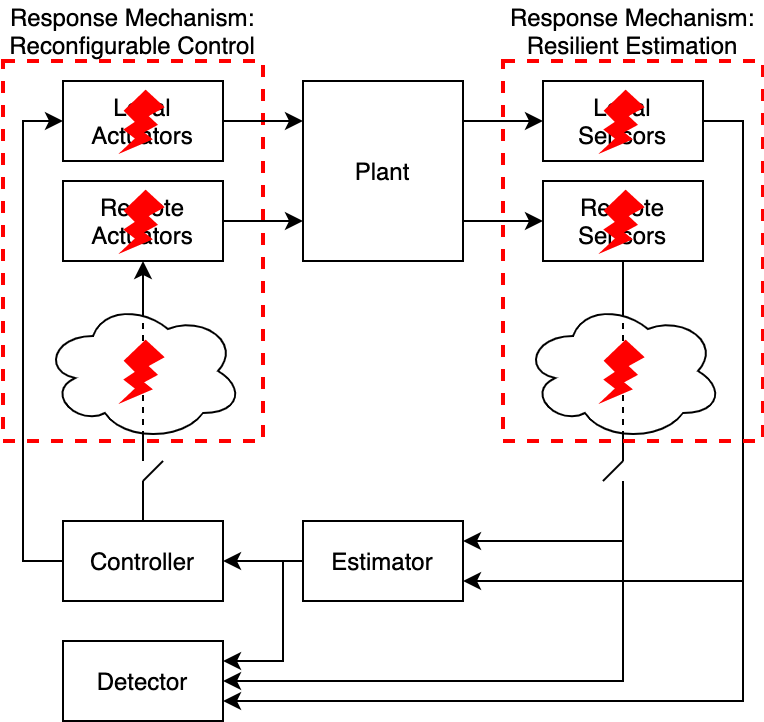}\label{fig:ReconfigurableControlResilientEstimation}}
\caption{Response mechanisms for attacks on the (a) computational, (b) communication, and (c) physical infrastructure.}
\label{fig:Response}
\end{figure}
When the control software is being attacked (Figure \ref{fig:SoftwareRejuvenation}), software rejuvenation can be used to achieve resilience. When the communication network is being attacked (Figure \ref{fig:OverlayNetworks}), overlay networks can be used to achieve resilience. When the physical actuators and/or the actuator communication channels are being attacked (Figure \ref{fig:ReconfigurableControlResilientEstimation}), reconfigurable control can be used to achieve resilience, and when the physical sensors and/or the sensor communication channels are being attacked (Figure \ref{fig:ReconfigurableControlResilientEstimation}), resilient estimation can be used to achieve resilience.

Software rejuvenation is an established method for dealing with so-called software aging in traditional computing systems, or failures that occur when a running program encounters a state that was not anticipated when the software was designed \cite{huang1995software,cotroneo2014survey}. In CPSs, software rejuvenation consists of periodically refreshing the run-time system with a secure and trusted copy of the control software, recovering the system from attacks during time periods called secure execution intervals when the control software can be trusted \cite{arroyo2017fired,abdi2018guaranteed,abdi2018preserving,arroyosecured}. The inertia in the physical dynamics of the system can be leveraged to keep the system safe while the control software is being refreshed, using tools from system theory to decide when the software should be refreshed \cite{romagnoli2019design,romagnoli2019safety,romagnoli2020robust,griffioen2019secure,romagnoli2020software,griffioen2020secure}.

Overlay networks can be used to ensure CPS resilience against man-in-the-middle attacks and denial of service (DoS) attacks \cite{andersen2001resilient,germanus2010increasing,farris2004evaluation,khelil2010towards,deconinck2008robust,benson2013improving}. This is accomplished by routing around malicious nodes and links in a network, specifying the pathway that is used each time data is sent between the controller and the plant. This functionality can also be achieved using software defined networking or source routing. Reachability analysis can be used to ensure that the safety constraints are not violated, even when a large portion of the pathways across the network have been compromised by adversaries \cite{griffioen2021resilient}.

Reconfigurable control \cite{zhang2008bibliographical,lunze2008reconfigurable} and resilient estimation \cite{fawzi2014secure,pajic2016attack} utilize functional redundancy in the system, ensuring resilience when up to a particular subset of actuator commands and sensor measurements are compromised. Reconfigurable control adapts the control policy to use those actuators that have not been compromised to preserve functionality, while resilient estimation constructs an accurate state estimate from those sensor measurements that have not yet been compromised.

\subsection{Trust and Time}
Each of these detection and response mechanisms leverage two important degrees of freedom in securing CPSs: trust and time. First, there must exist a component of the system that serves as the root of trust. For physical watermarking and the moving target defense this component is the pseudo-random number generator, for software rejuvenation this is the secure hypervisor, for overlay networks this is the set of uncompromised pathways, and for reconfigurable control and resilient estimation this is the set of uncompromised actuators or sensors. Second, the time periods where certain components and information may be trusted need to be identified and used to the defender’s advantage. For software rejuvenation this time period corresponds to the secure execution interval when the software can be trusted, for overlay networks this is the set of times where data is sent over uncompromised pathways, and for reconfigurable control and resilient estimation this is the set of times where a subset of the actuators or sensors can be trusted. Some approaches have leveraged decentralized event-triggered control to design protocols that minimize the amount of time components are untrusted or potentially compromised \cite{griffioen2021reducing,griffioen2020decentralized}.

In designing resilient detection and response mechanisms, we want to create time periods where we can trust certain components and information and use those components and information to recover the system from any potential attacks. Consequently, this article seeks to analyze those periods of time, investigating how long the components need to be restored to full functionality to ensure safety against stealthy attacks. We investigate how long the system can tolerate all the components malfunctioning or being manipulated by a stealthy adversary before restoration of those components must take place. This analysis will aid system operators as they design particular response strategies, no matter what the response strategy is. We provide sufficient conditions that the response strategy should meet in order to ensure the probabilistic safety of the overall system against stealthy attacks.

In summary, then, this article answers the following question: for a system designer who wants to ensure that the state remains within the safety constraints with a particular probability, what is a conservative estimate of the maximum amount of time the system can handle being under a stealthy attack? The response strategy should then be designed to ensure that the system can never be under stealthy attack for any longer than this amount of time.

%%%%%%%%%%%%%%%%%%%%%%%%%%%%%%%%%%%%%%%%%%%%%%%%%%%%%%%%%
\section{Problem Formulation}

\subsection{System Model}
To begin, we introduce the model for the system under consideration. We model the CPS as a discrete time linear time invariant system where the system dynamics are given by
\begin{align}
\label{StateDynamics}
x_{k+1} &=
\begin{cases}
Ax_k + Bu_k + w_k & \text{normal operation,} \\
Ax_k + B(u_k+u_k^a) + w_k & \text{attack,}
\end{cases} \\
\label{SensorDynamics}
y_k &=
\begin{cases}
Cx_k + v_k & \text{normal operation,} \\
Cx_k + y_k^a + v_k & \text{attack,}
\end{cases}
\end{align}
where normal operation is guaranteed when the system operates securely, free from attacks. Here $x_k\in\mathbb{R}^n$ represents the system state at time step $k$, $y_k\in\mathbb{R}^m$ denotes the sensor measurements, and $u_k\in\mathbb{R}^\ell$ is the control input vector. To capture uncertainty we consider independent and identically distributed (i.i.d.) Gaussian process noise $w_k\sim\mathcal{N}(0,Q)$ and i.i.d. Gaussian sensor noise $v_k\sim\mathcal{N}(0,R)$. We assume that $(A,B)$ and $(A,Q^\frac{1}{2})$ are stabilizable, $(A,C)$ is detectable, and $R\succ0$.

In this article, we consider an adversary who can perform integrity attacks, where the attacker is able to corrupt all of the inputs and outputs, denoted by adding a bias of $u_k^a$ and $y_k^a$ to the control inputs and sensor measurements, respectively \cite{cardenas2008secure,teixeira2015secure}. Without loss of generality, an attack is assumed to begin at time step $k=0$. This can be carried out for example through a memory corruption attack where the adversary corrupts the control software and injects malicious code \cite{szekeres2013sok}. Alternatively, this can be carried out through a man-in-the-middle attack where the adversary intercepts and modifies the data passing over the network. This can also be carried out by modifying the physical actuators and sensors themselves, where an adversary may change the settings of the actuators or the environment surrounding the system sensors. Motivated by the potential resources of sophisticated malicious adversaries, we consider an attacker who can modify all the inputs and outputs.

\begin{remark}
In the case where an adversary carries out a memory corruption attack, the results presented in this article are only applicable when the state estimator and detector are not corrupted under attack. This can be obtained, for instance, by implementing them within a secure hypervisor. However, no such qualification needs to be made if the adversary carries out the attack through other means, including man-in-the-middle attacks and attacks on the physical actuators and sensors themselves.
\end{remark}

We will also consider the possibility that an attacker has detailed system knowledge, including knowledge of the plant as well as the ability to read all of the inputs and outputs. This can be obtained through knowledge of the first principles which govern the dynamics or through system identification stemming from passive observation of the inputs and outputs. This knowledge, when combined with the ability to corrupt all the inputs and outputs, can lead to powerful undetectable attacks \cite{smith2015covert}. For example, an attacker can subtract his or her effect on the system dynamics by choosing an arbitrary sequence of attack inputs $\{u_k^a\}$ and then constructing the sensor measurement bias according to
\begin{equation}
y_k^a = -Cx_k^a, \quad x_{k+1}^a = Ax_k^a + Bu_k^a, \quad x_0^a = 0.
\end{equation}
In this way, an attacker can arbitrarily perturb the system along the controllable subspace $(A,B)$. Due to the linearity of the system, the probability distribution of the outputs under such an attack is identical to the probability distribution of the outputs under normal operation. Consequently, no standard bad data detector can recognize this adversarial behavior, and as a result this behavior is perfectly stealthy. In this article, we will analyze how resilient a system is to these stealthy attacks.

\subsection{State Estimation}
A Kalman filter is a linear estimator which can be used to compute the minimum mean squared error state estimate $\hat{x}_{k|k}$ given the set of previous measurements up through $y_k$. We assume that the system has been running for a long time so that the Kalman filter has converged to a fixed gain linear estimator given by
\begin{align}
\label{APrioriEstimate}
\hat{x}_{k+1|k} &= A\hat{x}_{k|k} + Bu_k, \\
\label{APosterioriEstimate}
\hat{x}_{k|k} &= \hat{x}_{k|k-1} + K(y_k-C\hat{x}_{k|k-1}),
\end{align}
where $\hat{x}_{k+1|k}$ is the \emph{a priori} state estimate and $\hat{x}_{k|k}$ is the \emph{a posteriori} state estimate. The Kalman gain $K$ and error covariance matrix $P$ are given by
\begin{align}
\label{KalmanGain}
K &= PC^T(CPC^T+R)^{-1}, \\
\label{CovarianceMatrix}
P &= APA^T-APC^T(CPC^T+R)^{-1}CPA^T + Q.
\end{align}
The control input is given by $u_k=L\hat{x}_{k|k}$, where $L$ is the control gain matrix. If the operator wants to minimize a quadratic function of the states and inputs according to
\begin{equation}
J = \lim_{\mathcal{T}\to\infty}\frac{1}{\mathcal{T}+1}\mathbb{E}\left[\sum_{k=0}^\mathcal{T}x_k^TWx_k+u_k^TVu_k\right],
\end{equation}
where $W$ and $V$ are positive definite matrices that respectively denote the relative cost of each state and control input, then the optimal $L$ is given by
\begin{align}
L &= -(B^TSB+V)^{-1}B^TSA, \\
\label{RiccatiEquation}
S &= A^TSA + W - A^TSB(B^TSB+V)^{-1}B^TSA,
\end{align}
where $S$ satisfies the Riccati equation in \eqref{RiccatiEquation} \cite{kumar2015stochastic}.

The residue $z_k$ is a function of the \emph{a priori} state estimate and represents the difference between the observed and expected value of the sensor measurements, given by
\begin{equation}
z_k\triangleq y_k-C\hat{x}_{k|k-1}.
\end{equation}
Depending on whether or not the system is under attack, the residue takes on different values, given by
\begin{equation}
z_k =
\begin{cases}
z_k^n & \text{normal operation,} \\
z_k^n+\Delta z_k & \text{attack,}
\end{cases}
\end{equation}
where $z_k^n$ represents the residue under normal operation which follows a normal distribution according to $z_k^n\sim\mathcal{N}(0,\Sigma)$ with $\Sigma\triangleq CPC^T+R$. When the system is under attack, the cumulative set of input and output biases $\{u_{0:k-1}^a,y_{0:k}^a\}$ results in the attacker exerting a bias $\Delta z_k$ on the residue, given by
\begin{equation}
\Delta z_k = y_k^a + \sum_{j=0}^{k-1}C(A(I-KC))^{k-1-j}(Bu_j^a-AKy_j^a).
\end{equation}

According to the system model in \eqref{StateDynamics}-\eqref{SensorDynamics} and the Kalman filter in \eqref{APrioriEstimate}-\eqref{APosterioriEstimate}, the error dynamics are given by
\begin{equation}
\medmuskip=3.7mu
\thinmuskip=3.7mu
\thickmuskip=3.7mu
e_{k+1|k+1} =
\begin{cases}
Ae_{k|k} + w_k - Kz_{k+1}^n \hspace{1.28cm}\text{normal operation,} \\
Ae_{k|k} + w_k - K(z_{k+1}^n+\Delta z_{k+1}) + Bu_k^a \text{ attack,}
\end{cases}
\end{equation}
where $e_{k|k}\triangleq x_k-\hat{x}_{k|k}$ is the \emph{a posteriori} state estimation error. The overall system dynamics are then given by
\begin{equation}
\label{OverallDynamics}
\bar{x}_{k+1} =
\begin{cases}
\mathcal{A}\bar{x}_k + \mathcal{I}w_k + \mathcal{K}z_{k+1}^n \hspace{1.25cm}\text{normal operation,} \\
\mathcal{A}\bar{x}_k + \mathcal{I}w_k + \mathcal{K}(z_{k+1}^n+\Delta z_{k+1}) + \mathcal{B}u_k^a \text{ attack,}
\end{cases}
\end{equation}
where $\bar{x}_k\triangleq\begin{bmatrix}x_k^T&e_{k|k}^T\end{bmatrix}^T$,
\begin{equation*}
\mathcal{A} \triangleq
\begin{bmatrix}
A+BL&-BL\\
0&A
\end{bmatrix}, ~
\mathcal{I} \triangleq
\begin{bmatrix}
I\\
I
\end{bmatrix}, ~
\mathcal{K} \triangleq
\begin{bmatrix}
0\\
-K
\end{bmatrix}, ~
\mathcal{B} \triangleq
\begin{bmatrix}
B\\
B
\end{bmatrix}.
\end{equation*}

\subsection{Detection}
To detect attacks on the CPS, a residue-based detector called the $\chi^2$ detector is utilized and is given by
\begin{equation}
g(z_{k-T+1:k}) = \sum_{i=k-T+1}^kz_i^T\Sigma^{-1}z_i \mathop{\gtrless}_{\mathcal{H}_0}^{\mathcal{H}_1} \eta,
\end{equation}
where $T$ represents the detector window that considers past measurements, and the detection statistic $g$ follows a $\chi^2$ distribution with $mT$ degrees of freedom under normal operation. The $\chi^2$ detector attempts to exploit this fact by testing to see if the residues follow the correct distribution. Here, $\eta$ represents the threshold of the bad data detector, $\mathcal{H}_0$ is the null hypothesis that represents normal system operation, and $\mathcal{H}_1$ is the alternative hypothesis that denotes that the system has deviated from normal operation. Measurements that are in close agreement with expected values generate small detection statistics and thus raise no alarm. Large deviations between measured and expected behavior will lead to a large detection statistic, thus causing an alarm.

When the system is under attack, the detection statistic $g$ follows a noncentral $\chi^2$ distribution, as shown in Lemmas \ref{NoncentralChiSquaredLemma} and \ref{DetectionStatisticDistributionLemma}.
\begin{lemma}
\label{NoncentralChiSquaredLemma}
Let $a\in\mathbb{R}^{n_a}\sim\mathcal{N}(\mu,\Omega)$. Then $a^T\Omega^{-1}a\sim\chi^2\left(n_a,\mu^T\Omega^{-1}\mu\right)$, where $n_a$ represents the degrees of freedom and $\mu^T\Omega^{-1}\mu$ is the noncentrality parameter.
\end{lemma}
\begin{proof}
Let $b\triangleq Y^{-1}a$, where $a\in\mathbb{R}^{n_a}\sim\mathcal{N}(\mu,\Omega)$ and $YY^T=\Omega$. Then $b\sim\mathcal{N}(Y^{-1}\mu,I)$, and
\begin{equation*}
a^T\Omega^{-1}a=b^Tb=\sum_{i=1}^{n_a}b_i^2
\end{equation*}
is distributed according to the noncentral $\chi^2$ distribution with $n_a$ degrees of freedom and a noncentrality parameter given by $\mu^T\Omega^{-1}\mu$.
\end{proof}
\begin{lemma}
\label{DetectionStatisticDistributionLemma}
When the system is under attack,
\begin{equation}
g(z_{k-T+1:k}) \sim \chi^2\left(mT,\sum_{i=k-T+1}^k\Delta z_i^T\Sigma^{-1}\Delta z_i\right),
\end{equation}
\begin{equation}
\label{ExpectedValue}
\mathbb{E}[g(z_{k-T+1:k})] = mT + \sum_{i=k-T+1}^k\Delta z_i^T\Sigma^{-1}\Delta z_i,
\end{equation}
where $mT$ represents the degrees of freedom and $\sum_{i=k-T+1}^k\Delta z_i^T\Sigma^{-1}\Delta z_i$ is the noncentrality parameter.
\end{lemma}
\begin{proof}
When the system is under attack, $z_k\sim\mathcal{N}(\Delta z_k,\Sigma)$. Then
\begin{equation*}
z_k^T\Sigma^{-1}z_k \sim \chi^2\left(m,\Delta z_k^T\Sigma^{-1}\Delta z_k\right)
\end{equation*}
according to Lemma \ref{NoncentralChiSquaredLemma}. Extending the result of Lemma \ref{NoncentralChiSquaredLemma} to
\begin{equation*}
g(z_{k-T+1:k})=
\begin{bmatrix} z_k \\ \vdots \\ z_{k-T+1} \end{bmatrix}^T
\begin{bmatrix}
\Sigma & \cdots & 0 \\
\vdots & \ddots & \vdots \\
0 & \cdots & \Sigma
\end{bmatrix}^{-1}
\begin{bmatrix} z_k \\ \vdots \\ z_{k-T+1} \end{bmatrix}
\end{equation*}
yields that $g(z_{k-T+1:k})$ is distributed according to the noncentral $\chi^2$ distribution with $mT$ degrees of freedom and a noncentrality parameter given by $\sum_{i=k-T+1}^k\Delta z_i^T\Sigma^{-1}\Delta z_i$. The expected value then follows directly from the properties of the noncentral $\chi^2$ distribution.
\end{proof}

\subsection{Saturation Limits}
Due to the physical saturation limits of the actuators, the vector of applied control inputs $\bar{u}_k$ will lie within the region $\bar{u}_k\in\mathcal{U}$, given by
\begin{equation}
\mathcal{U}\triangleq\left\{\bar{u}_k\middle|\bar{u}_k^TU\bar{u}_k\leq1\right\},
\end{equation}
where $U\succ0$ and
\begin{equation}
\bar{u}_k \triangleq
\begin{cases}
u_k & \text{normal operation,} \\
u_k+u_k^a & \text{attack.}
\end{cases}
\end{equation}
We assume that under normal operation the probability of the designed control inputs $u_k=L\hat{x}_{k|k}$ lying outside these saturation limits is negligible so that the system operates within the region $\bar{x}_k\in\mathcal{X}$, given by
\begin{equation}
\mathcal{X}\triangleq\left\{\bar{x}_k\middle|\bar{x}_k^T\bar{L}^TU\bar{L}\bar{x}_k\leq1\right\},
\end{equation}
where $\bar{L}\triangleq\begin{bmatrix}L&-L\end{bmatrix}$. When the system is under attack, these saturation limits will constrain the set of possible control input biases an attacker can exert on the system so that $u_k^a$ will lie within the region $u_k^a\in\mathcal{U}^a(\bar{x}_k)$, given by
\begin{equation}
\mathcal{U}^a(\bar{x}_k)\triangleq\left\{u_k^a\middle|(u_k^a+\bar{L}\bar{x}_k)^TU(u_k^a+\bar{L}\bar{x}_k)\leq1\right\}.
\end{equation}

%%%%%%%%%%%%%%%%%%%%%%%%%%%%%%%%%%%%%%%%%%%%%%%%%%%%%%%%%
\section{Stealthy Adversaries}
We would like to ensure resilience against stealthy adversaries, which we define as adversaries that can bypass the $\chi^2$ detector. This definition is expressed as follows.
\begin{definition}
\label{StealthyDefinition}
A $(p_d,k')$-stealthy adversary designs $\Delta z_{0:k'}$ such that $p(g(z_{k-T+1:k})>\eta|\mathcal{H}_1)\leq p_d$ $\forall k\in\{0,\cdots,k'\}$.
\end{definition}
In other words, a $(p_d,k')$-stealthy adversary designs the biases on the residues such that the probability of detection is less than or equal to $p_d$ through time step $k'$.

\subsection{$(p_d,k')$-Stealthy Bias Set}
To analyze the resilience of the system against stealthy adversaries, we first quantify the set of biases an adversary can exert on the residues that bypass the $\chi^2$ detector. Theorem \ref{BiasSetTheorem} describes this set.
\begin{theorem}
\label{BiasSetTheorem}
The full set of biases $\Delta z_{0:k'}$ that are $(p_d,k')$-stealthy is given by
\begin{equation}
\label{StealthyBiasSet}
Z_{0:k'}\triangleq\bigcap_{k=0}^{k'}\left\{\Delta z_{0:k'}\middle|\sum_{i=k-T+1}^k\Delta z_i^T\Sigma^{-1}\Delta z_i\leq\bar{\lambda}\right\},
\end{equation}
where
\begin{equation}
\label{NonCentralityMaximization}
\bar{\lambda}\triangleq\max_\lambda\lambda\quad\text{s.t.}~Q_{\frac{1}{2}mT}\left(\sqrt{\lambda},\sqrt{\eta}\right)\leq p_d,
\end{equation}
and $Q_{\frac{1}{2}mT}(\sqrt{\lambda},\sqrt{\eta})$ is the generalized Marcum $Q$-function, given by
\begin{equation}
Q_{\frac{1}{2}mT}(\sqrt{\lambda},\sqrt{\eta}) \triangleq 1-e^{-\frac{1}{2}\lambda}\sum_{j=0}^\infty\frac{(\frac{1}{2}\lambda)^j}{j!}\frac{\int_0^{\frac{1}{2}\eta}t^{\frac{1}{2}mT+j-1}e^{-t}dt}{\int_0^\infty t^{\frac{1}{2}mT+j-1}e^{-t}dt}.
\end{equation}
\end{theorem}
\begin{proof}
Lemma \ref{DetectionStatisticDistributionLemma} states that when the system is under attack, the detection statistic follows a noncentral $\chi^2$ distribution with a noncentrality parameter that is a quadratic function of the set of biases $\Delta z_{k-T+1:k}$. Consequently, the full set of biases that are $(p_d,k')$-stealthy is given by maximizing the noncentrality parameter subject to the conditions of $(p_d,k')$-stealthiness provided in Definition \ref{StealthyDefinition}. These conditions are set forth in the constraint of \eqref{NonCentralityMaximization} and the intersection of the sets in \eqref{StealthyBiasSet}, where the relationship between the generalized Marcum $Q$-function and the cumulative distribution function of the noncentral $\chi^2$ distribution is given by
\begin{equation*}
Q_{\frac{1}{2}mT}\left(\sqrt{\lambda},\sqrt{\eta}\right) = p\left(g(z_{k-T+1:k})>\eta\middle|\mathcal{H}_1\right).
\end{equation*}
Letting $\lambda$ represent the noncentrality parameter, \eqref{NonCentralityMaximization} maximizes this parameter and uses this maximum value to define the set in \eqref{StealthyBiasSet}.
\end{proof}

Lemmas \ref{MarcumQLemma} and \ref{MarcumQCorollary} show that the optimization problem in \eqref{NonCentralityMaximization} can be solved with a unique solution for $\bar{\lambda}$, allowing us to quantify the full set of $(p_d,k')$-stealthy biases $Z_{0:k'}$.
\begin{lemma}[\cite{sun2010monotonicity}]
\label{MarcumQLemma}
The generalized Marcum $Q$-function $Q_{\frac{1}{2}mT}(\sqrt{\lambda},\sqrt{\eta})$ is strictly increasing in $\frac{1}{2}mT$ and $\sqrt{\lambda}$ for all $\sqrt{\lambda}\geq0$, $\sqrt{\eta}>0$, and $\frac{1}{2}mT>0$. It is strictly decreasing in $\sqrt{\eta}$ for all $\sqrt{\lambda}\geq0$, $\sqrt{\eta}\geq0$, and $\frac{1}{2}mT>0$.
\end{lemma}
\begin{lemma}
\label{MarcumQCorollary}
The solution to the optimization problem in \eqref{NonCentralityMaximization} is
\begin{equation}
\label{NonCentralityMaxSolution}
\bar{\lambda}\triangleq\lambda\quad\text{s.t.}~Q_{\frac{1}{2}mT}\left(\sqrt{\lambda},\sqrt{\eta}\right)=p_d.
\end{equation}
\end{lemma}
\begin{proof}
Follows directly from the fact that according to Lemma \ref{MarcumQLemma} the constraint in \eqref{NonCentralityMaximization}, given by the generalized Marcum $Q$-function, is strictly increasing in $\lambda$ for all $\lambda\geq0$, $\eta>0$, and $mT>0$.
\end{proof}
In other words, Lemma \ref{MarcumQCorollary} allows us to compute an exact value for $\bar{\lambda}$ by successively increasing $\lambda$ until $Q_{\frac{1}{2}mT}\left(\sqrt{\lambda},\sqrt{\eta}\right)=p_d$.

Having now exactly quantified the set of $(p_d,k')$-stealthy biases $Z_{0:k'}$, Theorem \ref{BiasSetTheoremProjection} projects this set onto $\Delta z_k$ to decouple the correlation of the biases over time, in turn over-approximating the set of $(p_d,k')$-stealthy biases.
\begin{theorem}
\label{BiasSetTheoremProjection}
If an adversary remains $(p_d,k')$-stealthy, then $\forall k\in\{0,\cdots,k'\}$, $\Delta z_k$ will lie within the region $\Delta z_k\in\mathcal{Z}$, given by
\begin{equation}
\mathcal{Z}\triangleq\left\{\Delta z_k\middle|\Delta z_k^T\Sigma^{-1}\Delta z_k\leq\bar{\lambda}\right\}.
\end{equation}
\end{theorem}
\begin{proof}
Follows directly from projecting the set $Z_{0:k'}$ onto $\Delta z_k$. Projecting $\{\Delta z_{0:k'}|\sum_{i=j-T+1}^j\Delta z_i^T\Sigma^{-1}\Delta z_i\leq\bar{\lambda}\}$ onto $\Delta z_k$ yields $\Delta z_k\in\mathcal{Z}$ $\forall j=\{k,\cdots,k+T-1\}$ and $\Delta z_k\in\mathbb{R}^m$ otherwise. Taking the intersection of all these projections over $j=\{0,\cdots,k'\}$ yields $\mathcal{Z}$.
\end{proof}
\begin{remark}
Note that Theorem \ref{BiasSetTheoremProjection} states that for a $(p_d,k')$-stealthy adversary, $\Delta z_k\in\mathcal{Z}$ $\forall k\in\{0,\cdots,k'\}$. Since $\mathcal{Z}$ is time-invariant, $\Delta z_k$ will still lie within this region as the attacker's stealthy time horizon $k'\to\infty$. Consequently, the region $\mathcal{Z}$ is an over-approximation of the set of biases that are $p_d$-stealthy, where $p_d$-stealthiness is defined as follows.
\end{remark}
\begin{definition}
A $p_d$-stealthy adversary designs $\Delta z_{0:k'}$ such that $p(g(z_{k-T+1:k})>\eta|\mathcal{H}_1)\leq p_d$ $\forall k$.
\end{definition}
In other words, a $p_d$-stealthy adversary designs the biases on the residues such that the probability of detection is less than or equal to $p_d$ for all time.

\subsection{Alternate $(p_d,k')$-Stealthy Bias Set}
We can quantify an alternate set of biases that bypass the $\chi^2$ detector. While the relationship between the stealthy set of biases and the parameters $p_d$, $\eta$, and $T$ is not directly seen in \eqref{StealthyBiasSet}, this alternate stealthy bias set more directly depicts this relationship. Lemma \ref{DetectionBoundLemma} and Theorem \ref{AlternateBiasSetTheorem} describe this set.
\begin{lemma}
\label{DetectionBoundLemma}
The probability of detection is upper bounded according to
\begin{equation}
\medmuskip=2.17mu
\thinmuskip=2.17mu
\thickmuskip=2.17mu
p\left(g(z_{k-T+1:k})\geq\eta\middle|\mathcal{H}_1\right) \leq \frac{1}{\eta}\left(mT + \sum_{i=k-T+1}^k\Delta z_i^T\Sigma^{-1}\Delta z_i\right).
\end{equation}
\end{lemma}
\begin{proof}
According to the Markov inequality and the expected value of the noncentral $\chi^2$ distribution given in \eqref{ExpectedValue},
\begin{equation*}
\begin{split}
&p\left(g(z_{k-T+1:k})\geq\eta\middle|\mathcal{H}_1\right) \leq \mathbb{E}\left[g(z_{k-T+1:k})\middle|\mathcal{H}_1\right]/\eta \\
& \hspace{3.1cm} = \frac{1}{\eta}\left(mT + \sum_{i=k-T+1}^k\Delta z_i^T\Sigma^{-1}\Delta z_i\right).
\end{split}
\end{equation*}
\end{proof}
\begin{theorem}
\label{AlternateBiasSetTheorem}
If an adversary designs $\Delta z_{0:k'}$ so that they lie within the set
\begin{equation}
\label{AlternateBiasSet}
\bigcap_{k=0}^{k'}\left\{\Delta z_{0:k'}\middle|\sum_{i=k-T+1}^k\Delta z_i^T\Sigma^{-1}\Delta z_i\leq\eta p_d-mT\right\},
\end{equation}
then the adversary will remain $(p_d,k')$-stealthy.
\end{theorem}
\begin{proof}
\begin{equation*}
\medmuskip=1.22mu
\thinmuskip=1.22mu
\thickmuskip=1.22mu
\begin{split}
&\bigcap_{k=0}^{k'}\left\{\Delta z_{0:k'}\middle|\sum_{i=k-T+1}^k\Delta z_i^T\Sigma^{-1}\Delta z_i\leq\eta p_d-mT\right\} \\
&\subseteq \bigcap_{k=0}^{k'}\Bigg\{\left\{\Delta z_{0:k'}\middle|p_d\geq\frac{1}{\eta}\left(mT + \sum_{i=k-T+1}^k\Delta z_i^T\Sigma^{-1}\Delta z_i\right)\right\} \\
& \hspace{1.2cm}\oplus\Bigg\{\Delta z_{0:k'}\Bigg|p\left(g(z_{k-T+1:k})\geq\eta\middle|\mathcal{H}_1\right)\leq p_d~\land \\
& \hspace{2.7cm}p_d<\frac{1}{\eta}\left(mT + \sum_{i=k-T+1}^k\Delta z_i^T\Sigma^{-1}\Delta z_i\right)\Bigg\}\Bigg\} \\
&= \bigcap_{k=0}^{k'}\Bigg\{\Bigg\{\Delta z_{0:k'}\Bigg|p_d\geq\frac{1}{\eta}\left(mT + \sum_{i=k-T+1}^k\Delta z_i^T\Sigma^{-1}\Delta z_i\right)~\land \\
& p\left(g(z_{k-T+1:k})\geq\eta\middle|\mathcal{H}_1\right)\leq\frac{1}{\eta}\left(mT + \sum_{i=k-T+1}^k\Delta z_i^T\Sigma^{-1}\Delta z_i\right)\Bigg\} \\
& \hspace{1.2cm}\oplus\Bigg\{\Delta z_{0:k'}\Bigg|p\left(g(z_{k-T+1:k})\geq\eta\middle|\mathcal{H}_1\right)\leq p_d~\land \\
& \hspace{2.7cm}p_d<\frac{1}{\eta}\left(mT + \sum_{i=k-T+1}^k\Delta z_i^T\Sigma^{-1}\Delta z_i\right)\Bigg\}\Bigg\} \\
&= \bigcap_{k=0}^{k'}\Bigg\{\Delta z_{0:k'}\Bigg|p\left(g(z_{k-T+1:k})\geq\eta\middle|\mathcal{H}_1\right)\leq \\
& \hspace{2.0cm}\min\left(p_d,\frac{1}{\eta}\left(mT + \sum_{i=k-T+1}^k\Delta z_i^T\Sigma^{-1}\Delta z_i\right)\right)\Bigg\} \\
&= \bigcap_{k=0}^{k'}\Bigg\{\Delta z_{0:k'}\Bigg|p\left(g(z_{k-T+1:k})\geq\eta\middle|\mathcal{H}_1\right)\leq p_d~\land \\
& p\left(g(z_{k-T+1:k})\geq\eta\middle|\mathcal{H}_1\right)\leq\frac{1}{\eta}\left(mT + \sum_{i=k-T+1}^k\Delta z_i^T\Sigma^{-1}\Delta z_i\right)\Bigg\} \\
&= \bigcap_{k=0}^{k'}\left\{\Delta z_{0:k'}\middle|p\left(g(z_{k-T+1:k})\geq\eta\middle|\mathcal{H}_1\right)\leq p_d\right\} \\
&= \left\{\Delta z_{0:k'}\middle|p\left(g(z_{k-T+1:k})\geq\eta\middle|\mathcal{H}_1\right)\leq p_d~\forall k\in\{0,\cdots,k'\}\right\},
\end{split}
\end{equation*}
where the first and fourth equalities follow directly from Lemma \ref{DetectionBoundLemma}.
\end{proof}
Theorem \ref{AlternateBiasSetTheorem} sets forth a subset of the full set of $(p_d,k')$-stealthy biases in \eqref{AlternateBiasSet}, showing how this set constricts with increases in the time window $T$ of the detector and grows with increases in $p_d$ and the threshold $\eta$ of the detector. Consequently, this knowledge can be used when designing the detector to choose appropriate values for $\eta$ and $T$ that restrict the set of stealthy biases.

%%%%%%%%%%%%%%%%%%%%%%%%%%%%%%%%%%%%%%%%%%%%%%%%%%%%%%%%%
\section{Safety}
We now investigate the amount of time the overall system can be under stealthy attack and still remain safe with a particular probability. For the ease of analysis, we divide up the dynamics in \eqref{OverallDynamics} according to the bounded uncertainties $\{\Delta z_{k+1},u_k^a\}$ and the stochastic uncertainties $\{w_k,z_{k+1}^n\}$. Letting $\bar{x}_k^1$ and $\bar{x}_k^2$ be defined such that $\bar{x}_k^1+\bar{x}_k^2\triangleq\bar{x}_k$, the overall dynamics in \eqref{OverallDynamics} can be represented as
\begin{align}
\label{DisturbanceDynamics}
\bar{x}_{k+1}^1 &=
\begin{cases}
\mathcal{A}\bar{x}_k^1 & \text{normal operation,} \\
\mathcal{A}\bar{x}_k^1 + \mathcal{K}\Delta z_{k+1} + \mathcal{B}u_k^a & \text{attack,}
\end{cases} \\
\label{StochasticDynamics}
\bar{x}_{k+1}^2 &= \mathcal{A}\bar{x}_k^2 +
\underbrace{\begin{bmatrix}
\mathcal{I} & \mathcal{K}
\end{bmatrix}}_{\bar{\mathcal{K}}}
\underbrace{\begin{bmatrix}
w_k \\
z_{k+1}^n
\end{bmatrix}}_{\bar{w}_k}.
\end{align}
In the next two subsections, we provide sufficient conditions for ensuring that $\bar{x}_k^1$ and $\bar{x}_k^2$ respectively remain within particular ellipsoidal sets when under normal operation and attack.

\subsection{Bounded Uncertainties}
Theorems \ref{DisturbanceNormalOperationTheorem} and \ref{DisturbanceAttackTheorem} quantify how quickly the Lyapunov function $\bar{x}_k^{1^T}\mathcal{P}_1\bar{x}_k^1$ associated with the bounded uncertainties system in \eqref{DisturbanceDynamics} either decreases or increases when the system is under normal operation or attack.
\begin{theorem}
\label{DisturbanceNormalOperationTheorem}
For the system in \eqref{DisturbanceDynamics} under normal operation,
\begin{equation}
\gamma\mathcal{P}_1-\mathcal{A}^T\mathcal{P}_1\mathcal{A}\succeq0\iff\bar{x}_{k+1}^{1^T}\mathcal{P}_1\bar{x}_{k+1}^1\leq\gamma\bar{x}_k^{1^T}\mathcal{P}_1\bar{x}_k^1,
\end{equation}
where $\gamma\geq0$ and $\mathcal{P}_1$ is a symmetric positive definite Lyapunov matrix.
\end{theorem}
\begin{proof}
\begin{equation*}
\begin{split}
\gamma\mathcal{P}_1-\mathcal{A}^T\mathcal{P}_1\mathcal{A}\succeq0 &\iff \bar{x}_k^{1^T}\mathcal{A}^T\mathcal{P}_1\mathcal{A}\bar{x}_k^1\leq\gamma\bar{x}_k^{1^T}\mathcal{P}_1\bar{x}_k^1 \\
&\iff \bar{x}_{k+1}^{1^T}\mathcal{P}_1\bar{x}_{k+1}^1\leq\gamma\bar{x}_k^{1^T}\mathcal{P}_1\bar{x}_k^1
\end{split}
\end{equation*}
\end{proof}
Theorem \ref{DisturbanceNormalOperationTheorem} provides the condition guaranteeing that the Lyapunov function associated with the bounded uncertainties system in \eqref{DisturbanceDynamics} will change at a rate upper bounded by $\gamma$ when the system is under normal operation. To ensure that the system state converges towards the origin under normal operation, $\gamma\mathcal{P}_1-\mathcal{A}^T\mathcal{P}_1\mathcal{A}\succeq0$ should be satisfied with $\gamma\in[0,1)$ so that the Lyapunov function decreases over time.

\begin{theorem}
\label{DisturbanceAttackTheorem}
If $\exists\alpha_1\geq0$ such that
\begin{equation}
\label{QuadraticBoundednessLMI}
\begin{bmatrix}
\Gamma_{11} & \Gamma_{12} \\
\Gamma_{12}^T & \Gamma_{22}
\end{bmatrix}
\succeq 0,
\end{equation}
where
\begin{align}
\Gamma_{11}&\triangleq
\begin{bmatrix}
(\gamma_a-2\alpha_1)\mathcal{P}_1-\mathcal{A}^T\mathcal{P}_1\mathcal{A}+\alpha_1\bar{L}^TU\bar{L} & \alpha_1\bar{L}^TU\bar{L} \\
\alpha_1\bar{L}^TU\bar{L} & \alpha_1\bar{L}^TU\bar{L}
\end{bmatrix}, \nonumber\\
\Gamma_{12}&\triangleq
\begin{bmatrix}
-\mathcal{A}^T\mathcal{P}_1\mathcal{K} & \alpha_1\bar{L}^TU-\mathcal{A}^T\mathcal{P}_1\mathcal{B} \\
0 & \alpha_1\bar{L}^TU
\end{bmatrix}, \nonumber\\
\Gamma_{22}&\triangleq
\begin{bmatrix}
\frac{\alpha_1}{\bar{\lambda}}\Sigma^{-1}-\mathcal{K}^T\mathcal{P}_1\mathcal{K} & -\mathcal{K}^T\mathcal{P}_1\mathcal{B} \\
-\mathcal{B}^T\mathcal{P}_1\mathcal{K} & \alpha_1U-\mathcal{B}^T\mathcal{P}_1\mathcal{B}
\end{bmatrix}, \nonumber
\end{align}
then $\forall u_k^a\in\mathcal{U}^a(\bar{x}_k)$ and $\forall\Delta z_{k+1}\in\mathcal{Z}$,
\begin{equation}
\bar{x}_k^{1^T}\mathcal{P}_1\bar{x}_k^1\geq1 \implies \bar{x}_{k+1}^{1^T}\mathcal{P}_1\bar{x}_{k+1}^1\leq\gamma_a\bar{x}_k^{1^T}\mathcal{P}_1\bar{x}_k^1
\end{equation}
for system in \eqref{DisturbanceDynamics} when under attack, where $\gamma_a\geq0$ and $\mathcal{P}_1$ is a symmetric positive definite Lyapunov matrix.
\end{theorem}
\begin{proof}
By using the S-procedure \cite{boyd1994linear}, \eqref{QuadraticBoundednessLMI} is equivalent to
\begin{equation*}
\begin{split}
&\xi_k^T
\begin{bmatrix}
-2\mathcal{P}_1+\bar{L}^TU\bar{L} & \bar{L}^TU\bar{L} & 0 & \bar{L}^TU \\
\bar{L}^TU\bar{L} & \bar{L}^TU\bar{L} & 0 & \bar{L}^TU \\
0 & 0 & \frac{1}{\bar{\lambda}}\Sigma^{-1} & 0 \\
U\bar{L} & U\bar{L} & 0 & U
\end{bmatrix}
\xi_k\leq0\implies\\
&\hspace{0.9cm}\xi_k^T
\begin{bmatrix}
\mathcal{A}^T\mathcal{P}_1\mathcal{A}-\gamma_a\mathcal{P}_1 & 0 & \mathcal{A}^T\mathcal{P}_1\mathcal{K} & \mathcal{A}^T\mathcal{P}_1\mathcal{B} \\
0 & 0 & 0 & 0 \\
\mathcal{K}^T\mathcal{P}_1\mathcal{A} & 0 & \mathcal{K}^T\mathcal{P}_1\mathcal{K} & \mathcal{K}^T\mathcal{P}_1\mathcal{B} \\
\mathcal{B}^T\mathcal{P}_1\mathcal{A} & 0 & \mathcal{B}^T\mathcal{P}_1\mathcal{K} & \mathcal{B}^T\mathcal{P}_1\mathcal{B}
\end{bmatrix}
\xi_k\leq0,
\end{split}
\end{equation*}
where $\xi_k\triangleq\begin{bmatrix}\bar{x}_k^{1^T}&\bar{x}_k^{2^T}&\Delta z_{k+1}^T&u_k^{a^T}\end{bmatrix}^T$. This in turn is equivalent to
\begin{equation}
\label{Implication1}
\begin{split}
& -2\bar{x}_k^{1^T}\mathcal{P}_1\bar{x}_k^1+\Delta z_{k+1}^T\frac{1}{\bar{\lambda}}\Sigma^{-1}\Delta z_{k+1}\\
&\hspace{0.3cm}+(u_k^a+\bar{L}\bar{x}_k^1+\bar{L}\bar{x}_k^2)^TU(u_k^a+\bar{L}\bar{x}_k^1+\bar{L}\bar{x}_k^2)\leq0 \\
&\hspace{2.9cm}\implies \bar{x}_{k+1}^{1^T}\mathcal{P}_1\bar{x}_{k+1}^1\leq\gamma_a\bar{x}_k^{1^T}\mathcal{P}_1\bar{x}_k^1
\end{split}
\end{equation}
for system in \eqref{DisturbanceDynamics} when under attack. Note that
\begin{equation}
\label{Implication2}
\begin{split}
&\begin{Bmatrix}
\bar{x}_k^{1^T}\mathcal{P}_1\bar{x}_k^1\geq1\\
\Delta z_{k+1}^T\Sigma^{-1}\Delta z_{k+1}\leq\bar{\lambda}\\
(u_k^a+\bar{L}\bar{x}_k)^TU(u_k^a+\bar{L}\bar{x}_k)\leq1
\end{Bmatrix}\implies \\
&\hspace{0.4cm}-2\bar{x}_k^{1^T}\mathcal{P}_1\bar{x}_k^1+\Delta z_{k+1}^T\frac{1}{\bar{\lambda}}\Sigma^{-1}\Delta z_{k+1} \\
&\hspace{0.7cm}+(u_k^a+\bar{L}\bar{x}_k^1+\bar{L}\bar{x}_k^2)^TU(u_k^a+\bar{L}\bar{x}_k^1+\bar{L}\bar{x}_k^2)\leq0. \\
\end{split}
\end{equation}
Taking \eqref{Implication1} and \eqref{Implication2} in conjunction with one another yields that $\forall u_k^a\in\mathcal{U}^a(\bar{x}_k)$ and $\forall\Delta z_{k+1}\in\mathcal{Z}$,
\begin{equation*}
\bar{x}_k^{1^T}\mathcal{P}_1\bar{x}_k^1\geq1 \implies \bar{x}_{k+1}^{1^T}\mathcal{P}_1\bar{x}_{k+1}^1\leq\gamma_a\bar{x}_k^{1^T}\mathcal{P}_1\bar{x}_k^1
\end{equation*}
for system in \eqref{DisturbanceDynamics} when under attack.
\end{proof}
Theorem \ref{DisturbanceAttackTheorem} provides a sufficient condition guaranteeing that the Lyapunov function associated with the bounded uncertainties system in \eqref{DisturbanceDynamics} will decrease at a rate of at least $\gamma_a$ if $\gamma_a\in[0,1)$ or will increase at a rate of at most $\gamma_a$ if $\gamma_a\geq1$ when the system is under $p_d$-stealthy attack, where the adversary exerts biases $u_k^a$ and $\Delta z_{k+1}$ within the saturation and stealthiness limits of $\mathcal{U}^a(\bar{x}_k)$ and $\mathcal{Z}$, respectively.

Having now quantified the rate at which the Lyapunov function decreases or increases for the system in \eqref{DisturbanceDynamics}, Theorem \ref{DisturbanceSetTheorem} quantifies the size of the ellipsoid in which $\bar{x}_k^1$ is guaranteed to lie. The size of this ellipsoid is a function of the Lyapunov rate parameters $\gamma$ and $\gamma_a$ as well as the amount of time the system has been vulnerable to attacks.
\begin{theorem}
\label{DisturbanceSetTheorem}
If $\gamma\mathcal{P}_1-\mathcal{A}^T\mathcal{P}_1\mathcal{A}\succeq0$, if $\exists\alpha_1\geq0$ such that \eqref{QuadraticBoundednessLMI} is satisfied, and if $\bar{x}_0^1\in\mathcal{E}_1(0,0)$, then
\begin{equation}
\medmuskip=2.51mu
\thinmuskip=2.51mu
\thickmuskip=2.51mu
\label{E1Definition}
\bar{x}_k^1\in\mathcal{E}_1(k,T_a)\triangleq\left\{\bar{x}_k^1\middle|\bar{x}_k^{1^T}\mathcal{P}_1\bar{x}_k^1\leq\max(1,\gamma^{k-T_a}\gamma_a^{T_a})\right\},
\end{equation}
where $T_a$ represents the total number of time steps the system has been under $p_d$-stealthy attack.
\end{theorem}
\begin{proof}
If $\gamma\mathcal{P}_1-\mathcal{A}^T\mathcal{P}_1\mathcal{A}\succeq0$ and if $\exists\alpha_1\geq0$ such that \eqref{QuadraticBoundednessLMI} is satisfied, then it follows directly from Theorems \ref{DisturbanceNormalOperationTheorem} and \ref{DisturbanceAttackTheorem} that
\begin{equation*}
\bar{x}_{k+1}^{1^T}\mathcal{P}_1\bar{x}_{k+1}^1 \leq
\begin{cases}
\gamma\bar{x}_k^{1^T}\mathcal{P}_1\bar{x}_k^1 & \text{normal operation} \\
\gamma_a\bar{x}_k^{1^T}\mathcal{P}_1\bar{x}_k^1 & \text{attack when } \bar{x}_k^{1^T}\mathcal{P}_1\bar{x}_k^1\geq1
\end{cases}
\end{equation*}
for the system in \eqref{DisturbanceDynamics}. Consequently,
\begin{equation*}
\bar{x}_k^{1^T}\mathcal{P}_1\bar{x}_k^1 \leq \max\left(1,\gamma^{k-T_a}\gamma_a^{T_a}\bar{x}_0^{1^T}\mathcal{P}_1\bar{x}_0^1\right).
\end{equation*}
If $\bar{x}_0^1\in\mathcal{E}_1(0,0)$, then $\max_{\bar{x}_0^1}\bar{x}_0^{1^T}\mathcal{P}_1\bar{x}_0^1=1$, implying that
\begin{equation*}
\bar{x}_k^{1^T}\mathcal{P}_1\bar{x}_k^1 \leq \max\left(1,\gamma^{k-T_a}\gamma_a^{T_a}\right).
\end{equation*}
\end{proof}
Theorem \ref{DisturbanceSetTheorem} shows that when $\gamma\in[0,1)$ and $\gamma_a\geq1$, the ellipsoid $\mathcal{E}_1(k,T_a)$ grows by a factor of $\gamma_a$ each time step the system is under attack and shrinks by a factor of $\gamma$ each time step the system is not under attack. Corollary \ref{VolumeCorollary} describes the volume of this ellipsoid with respect to $\gamma$, $\gamma_a$, $k$, and $T_a$.
\begin{corollary}
\label{VolumeCorollary}
When $\gamma^{k-T_a}\gamma_a^{T_a}>1$, the volume of $\mathcal{E}_1(k,T_a)$
\begin{itemize}
\item grows logarithmically with $\gamma$ and $\gamma_a$.
\item is inversely proportional to $k$ if $\gamma\in(0,1)$.
\item is proportional to $T_a$ if $\gamma_a>\gamma$.
\end{itemize}
When $\gamma^{k-T_a}\gamma_a^{T_a}\leq1$, the volume of $\mathcal{E}_1(k,T_a)$ is not dependent on $\gamma$, $\gamma_a$, $k$, or $T_a$.
\end{corollary}
\begin{proof}
The volume of $\mathcal{E}_1(k,T_a)$ is proportional to
\begin{equation*}
\begin{split}
&-\log\det\left(\frac{1}{\max(1,\gamma^{k-T_a}\gamma_a^{T_a})}\mathcal{P}_1\right) \\
&= -\log\left(\frac{1}{\max(1,\gamma^{k-T_a}\gamma_a^{T_a})^{2n}}\det(\mathcal{P}_1)\right) \\
&= 2n\log\left(\max(1,\gamma^{k-T_a}\gamma_a^{T_a})\right) - \log\det(\mathcal{P}_1) \\
&= \begin{cases}
2n\left((k-T_a)\log(\gamma)+T_a\log(\gamma_a)\right) - \log\det(\mathcal{P}_1) \\
\hspace{5.5cm}\text{if } \gamma^{k-T_a}\gamma_a^{T_a} > 1 \\
-\log\det(\mathcal{P}_1) \hspace{3.5cm}\text{if } \gamma^{k-T_a}\gamma_a^{T_a} \leq 1
\end{cases} \\
&= \begin{cases}
2n\left(k\log(\gamma)+T_a\log\left(\frac{\gamma_a}{\gamma}\right)\right) - \log\det(\mathcal{P}_1) \\
\hspace{5.5cm}\text{if } \gamma^{k-T_a}\gamma_a^{T_a} > 1 \\
-\log\det(\mathcal{P}_1) \hspace{3.5cm}\text{if } \gamma^{k-T_a}\gamma_a^{T_a} \leq 1.
\end{cases}
\end{split}
\end{equation*}
\end{proof}

\subsection{Stochastic Uncertainties}
Having quantified the size of the ellipsoid in which $\bar{x}_k^1$ lies, we now quantify the ellipsoid in which $\bar{x}_k^2$ lies with at least probability $p$ for the stochastic uncertainties system in \eqref{StochasticDynamics}. To do so, we leverage the notions of quadratic boundedness, robust positive invariance, and probabilistic positive invariance which are described in Definitions \ref{QuadraticBoundedness}, \ref{RobustPositiveInvariance}, and \ref{ProbabilisticPositiveInvariance}. These definitions, along with Lemmas \ref{RobustInvarianceTheorem} and \ref{RIStoPIS}, have been modified and adapted from \cite{alessandri2004estimation,hewing2018correspondence} in order to arrive at the result presented in Theorem \ref{ProbabilisticInvarianceTheorem}.
\begin{definition}[\cite{alessandri2004estimation}]
\label{QuadraticBoundedness}
Let $s_k\in\mathbb{R}^{n_s}$ and $d_k\in\mathbb{R}^{n_d}$ represent state and disturbance vectors, respectively, and let $D$ be a compact set. A system of the form
\begin{equation}
\label{QuadraticBoundSystem}
s_{k+1} = \bar{\mathcal{A}}s_k + \bar{\mathcal{B}}d_k
\end{equation}
is quadratically bounded with symmetric positive definite Lyapunov matrix $\mathcal{P}$ if and only if
\begin{equation}
s_k^T\mathcal{P}s_k\geq1\implies s_{k+1}^T\mathcal{P}s_{k+1}\leq s_k^T\mathcal{P}s_k ~ \forall d_k\in D.
\end{equation}
\end{definition}
\begin{definition}[\cite{alessandri2004estimation}]
\label{RobustPositiveInvariance}
The set $\mathcal{S}$ is a robustly positively invariant set for \eqref{QuadraticBoundSystem} if and only if $s_k\in \mathcal{S}$ implies that $s_{k+1}\in\mathcal{S}$ $\forall d_k\in D$.
\end{definition}
\begin{definition}[\cite{hewing2018correspondence}]
\label{ProbabilisticPositiveInvariance}
Let $d_k\sim\mathcal{Q}^d$ indicate that $d_k$ is a random variable of distribution $\mathcal{Q}^d$. The set $\mathcal{S}$ is a probabilistic positively invariant set of probability level $p$ for \eqref{QuadraticBoundSystem} with $d_k\sim\mathcal{Q}^d$ if and only if $s_0\in\mathcal{S}$ implies that Pr$(s_k\in\mathcal{S})\geq p$ $\forall k\geq0$.
\end{definition}

Given these definitions, Lemma \ref{RobustInvarianceTheorem} provides a condition that is equivalent to robust positive invariance, Lemma \ref{RIStoPIS} shows the correspondence between robust positive invariance and probabilistic positive invariance, and Lemma \ref{StochasticLMI} combines these results to provide a sufficient condition for probabilistic positive invariance.
\begin{lemma}[\cite{alessandri2004estimation}]
\label{RobustInvarianceTheorem}
Let $d_k\in D=\left\{d_k\middle|d_k^T\mathcal{D}^{-1}d_k\leq1,~\mathcal{D}\succ0\right\}$. $\mathcal{S}\triangleq\left\{s_k\middle|s_k^T\mathcal{P}s_k\leq1\right\}$ is a robustly positively invariant set for \eqref{QuadraticBoundSystem} and the system in \eqref{QuadraticBoundSystem} is quadratically bounded with symmetric positive definite Lyapunov matrix $\mathcal{P}$ if and only if $\exists\alpha\geq0$ such that
\begin{equation}
\begin{bmatrix}
(\alpha-1)\mathcal{P}+\bar{\mathcal{A}}^T\mathcal{P}\bar{\mathcal{A}} & \bar{\mathcal{A}}^T\mathcal{P}\bar{\mathcal{B}} \\
\bar{\mathcal{B}}^T\mathcal{P}\bar{\mathcal{A}} & \bar{\mathcal{B}}^T\mathcal{P}\bar{\mathcal{B}}-\alpha\mathcal{D}^{-1}
\end{bmatrix} \preceq 0.
\end{equation}
\end{lemma}
\begin{lemma}[\cite{hewing2018correspondence}]
\label{RIStoPIS}
If $\mathcal{S}$ is a robustly positively invariant set for \eqref{QuadraticBoundSystem} with
\begin{equation}
d_k\in D=\left\{d_k\middle|d_k^T\mathcal{D}^{-1}d_k\leq n_d/(1-p)\right\},
\end{equation}
then $\mathcal{S}$ is also a probabilistic positively invariant set of probability level $p$ for \eqref{QuadraticBoundSystem} with $d_k\sim\mathcal{Q}^d$, $\mathbb{E}[d_k]=0$, and Cov$(d_k)=\mathcal{D}$.
\end{lemma}
\begin{lemma}
\label{StochasticLMI}
Let $d_k\sim\mathcal{Q}^d$, $\mathbb{E}[d_k]=0$, and Cov$(d_k)=\mathcal{D}$. If $\exists\alpha\geq0$, $\mathcal{P}\succ0$ such that
\begin{equation}
\begin{bmatrix}
(\alpha-1)\mathcal{P}+\bar{\mathcal{A}}^T\mathcal{P}\bar{\mathcal{A}} & \bar{\mathcal{A}}^T\mathcal{P}\bar{\mathcal{B}} \\
\bar{\mathcal{B}}^T\mathcal{P}\bar{\mathcal{A}} & \bar{\mathcal{B}}^T\mathcal{P}\bar{\mathcal{B}}-\frac{\alpha(1-p)}{n_d}\mathcal{D}^{-1}
\end{bmatrix} \preceq 0,
\end{equation}
then $\mathcal{S}\triangleq\left\{s_k\middle|s_k^T\mathcal{P}s_k\leq1\right\}$ is a probabilistic positively invariant set of probability level $p$ for \eqref{QuadraticBoundSystem}.
\end{lemma}
\begin{proof}
Follows directly from Lemmas \ref{RobustInvarianceTheorem} and \ref{RIStoPIS}.
\end{proof}

Given these Lemmas, Theorem \ref{ProbabilisticInvarianceTheorem} sets forth the ellipsoid in which $\bar{x}_k^2$ lies with at least probability $p$ for the system in \eqref{StochasticDynamics}, and Corollary \ref{VolumeCorollary2} describes the volume of this ellipsoid with respect to $p$.
\begin{theorem}
\label{ProbabilisticInvarianceTheorem}
Let $\mathcal{R}\triangleq\text{BlkDiag}(Q,\Sigma)$ so that $\bar{w}_k\in\mathbb{R}^{n+m}\sim\mathcal{N}(0,\mathcal{R})$. If $\exists\alpha_2\geq0$, $\mathcal{P}_2\succ0$ such that
\begin{equation}
\label{ProbabilisticBoundednessLMI}
\begin{bmatrix}
(\alpha_2-1)\mathcal{P}_2+\mathcal{A}^T\mathcal{P}_2\mathcal{A} & \mathcal{A}^T\mathcal{P}_2\bar{\mathcal{K}} \\
\bar{\mathcal{K}}^T\mathcal{P}_2\mathcal{A} & \bar{\mathcal{K}}^T\mathcal{P}_2\bar{\mathcal{K}}-\frac{\alpha_2}{n+m}\mathcal{R}^{-1}
\end{bmatrix} \preceq 0,
\end{equation}
then
\begin{equation}
\mathcal{E}_2(p)\triangleq\left\{\bar{x}_k^2\middle|\bar{x}_k^{2^T}\mathcal{P}_2\bar{x}_k^2\leq1/(1-p)\right\}
\end{equation}
is a probabilistic positively invariant set of probability level $p$ for \eqref{StochasticDynamics}.
\end{theorem}
\begin{proof}
Applying Lemma \ref{StochasticLMI} to \eqref{StochasticDynamics} yields that if $\exists\alpha_2\geq0$, $\bar{\mathcal{P}}\succ0$ such that
\begin{equation*}
\begin{bmatrix}
(\alpha_2-1)\bar{\mathcal{P}}+\mathcal{A}^T\bar{\mathcal{P}}\mathcal{A} & \mathcal{A}^T\bar{\mathcal{P}}\bar{\mathcal{K}} \\
\bar{\mathcal{K}}^T\bar{\mathcal{P}}\mathcal{A} & \bar{\mathcal{K}}^T\bar{\mathcal{P}}\bar{\mathcal{K}}-\frac{\alpha_2(1-p)}{n+m}\mathcal{R}^{-1}
\end{bmatrix} \preceq 0,
\end{equation*}
then $\mathcal{E}_2(p)=\left\{\bar{x}_k^2\middle|\bar{x}_k^{2^T}\bar{\mathcal{P}}\bar{x}_k^2\leq1\right\}$ is a probabilistic positively invariant set of probability level $p$ for \eqref{StochasticDynamics}. Letting $\bar{\mathcal{P}}\triangleq(1-p)\mathcal{P}_2$ yields the desired result.
\end{proof}
\begin{corollary}
\label{VolumeCorollary2}
The volume of $\mathcal{E}_2(p)$ grows logarithmically with $p$.
\end{corollary}
\begin{proof}
The volume of $\mathcal{E}_2(p)$ is proportional to
\begin{equation*}
\begin{split}
-\log\det\left((1-p)\mathcal{P}_2\right) &= -\log\left((1-p)^{2n}\det(\mathcal{P}_2)\right) \\
&= -2n\log(1-p)-\log\det(\mathcal{P}_2).
\end{split}
\end{equation*}
\end{proof}

\subsection{Safety Set}
We have now quantified the ellipsoid $\mathcal{E}_1(k,T_a)$ in which $\bar{x}_k^1$ is guaranteed to lie as well as the ellipsoid $\mathcal{E}_2(p)$ in which $\bar{x}_k^2$ lies with at least probability $p$. Theorem \ref{SafetyTheorem} uses these results to quantify the region within which the overall state $\bar{x}_k$ lies, and this region can be used to determine whether or not the safety constraints have been violated.
\begin{theorem}
\label{SafetyTheorem}
If $\bar{x}_0\in\{\mathcal{E}_1(0,0)\oplus\mathcal{E}_2(p)\}$ and if $\exists\alpha_1,\alpha_2,\gamma,\gamma_a\geq0$ and $\mathcal{P}_1,\mathcal{P}_2\succ0$ such that $\gamma\mathcal{P}_1-\mathcal{A}^T\mathcal{P}_1\mathcal{A}\succeq0$, \eqref{QuadraticBoundednessLMI}, and \eqref{ProbabilisticBoundednessLMI} are satisfied, then
\begin{equation}
\text{Pr}\left(\bar{x}_k\in\left\{\mathcal{E}_1(k,T_a)\oplus\mathcal{E}_2(p)\right\}\right)\geq p ~ \forall k\geq0
\end{equation}
when a $p_d$-stealthy adversary attacks the system for a total of $T_a$ time steps.
\end{theorem}
\begin{proof}
If $\bar{x}_0\in\{\mathcal{E}_1(0,0)\oplus\mathcal{E}_2(p)\}$ and if $\exists\alpha_1,\alpha_2,\gamma,\gamma_a\geq0$ and $\mathcal{P}_1,\mathcal{P}_2\succ0$ such that $\gamma\mathcal{P}_1-\mathcal{A}^T\mathcal{P}_1\mathcal{A}\succeq0$, \eqref{QuadraticBoundednessLMI}, and \eqref{ProbabilisticBoundednessLMI} are satisfied, then it follows directly from Theorems \ref{DisturbanceSetTheorem} and \ref{ProbabilisticInvarianceTheorem} that
\begin{equation*}
\bar{x}_k^1\in\mathcal{E}_1(k,T_a), \quad \text{Pr}\left(\bar{x}_k^2\in\mathcal{E}_2(p)\right)\geq p ~ \forall k\geq0.
\end{equation*}
Since $\bar{x}_k=\bar{x}_k^1+\bar{x}_k^2$, $\bar{x}_k\in\{\mathcal{E}_1(k,T_a)\oplus\mathcal{E}_2(p)\}$ with probability greater than or equal to $p$ $\forall k\geq0$.
\end{proof}
Given a set of safe state constraints, Theorem \ref{SafetyTheorem} can be used to provide a conservative estimate of how long the system can remain under $p_d$-stealthy attack before the safe state constraints are violated. It provides a tool to analyze the response mechanism being used, providing the probability of safety given the number of time steps the system has been vulnerable to attack.

In executing this analysis, the variables $\gamma$, $\gamma_a$, $\mathcal{P}_1$, and $\mathcal{P}_2$ are all parameters that can be chosen by the system operator as long as $\gamma\mathcal{P}_1-\mathcal{A}^T\mathcal{P}_1\mathcal{A}\succeq0$, \eqref{QuadraticBoundednessLMI}, and \eqref{ProbabilisticBoundednessLMI} are satisfied. The meanings of these parameters are presented as follows:
\begin{itemize}
\item $\gamma$: The rate at which the Lyapunov function decreases ($\gamma\in[0,1)$) or increases ($\gamma\geq1$) under normal operation for the system in \eqref{DisturbanceDynamics}.
\item $\gamma_a$: The rate at which the Lyapunov function decreases ($\gamma_a\in[0,1)$) or increases ($\gamma_a\geq1$) under attack for the system in \eqref{DisturbanceDynamics}.
\item $\mathcal{P}_1$: The shape matrix for $\mathcal{E}_1(k,T_a)$.
\item $\mathcal{P}_2$: The shape matrix for $\mathcal{E}_2(p)$.
\end{itemize}
To minimize the rate parameters $\gamma$ and $\gamma_a$ as well as minimize the volume of the shape matrix for the ellipsoid $\mathcal{E}_1(k,T_a)$, the parameters $\gamma$, $\gamma_a$, and $\mathcal{P}_1$ can be designed according to the following optimization problem:
\begin{equation}
\label{Optimization1}
\begin{split}
&\argmin_{\alpha_1,\gamma,\gamma_a,\mathcal{P}_1}\omega_1\gamma+\omega_2\gamma_a-\omega_3\log\det{\mathcal{P}_1}\hspace{0.6cm}\text{s.t.}~\mathcal{P}_1\succ0, \\
&\hspace{0.3cm}\alpha_1,\gamma,\gamma_a\geq0,~\gamma\mathcal{P}_1-\mathcal{A}^T\mathcal{P}_1\mathcal{A}\succeq0,~\text{\eqref{QuadraticBoundednessLMI} is satisfied,}
\end{split}
\end{equation}
where $\omega_1$, $\omega_2$, and $\omega_3$ are nonnegative constants weighting the rate at which the Lyapunov function for \eqref{DisturbanceDynamics} decreases/increases under normal operation, the the rate at which the Lyapunov function for \eqref{DisturbanceDynamics} decreases/increases under attack, and the volume of the shape matrix for $\mathcal{E}_1(k,T_a)$, respectively. To minimize the volume of the ellipsoid $\mathcal{E}_2(p)$, the parameter $\mathcal{P}_2$ can be designed according to the following optimization problem:
\begin{equation}
\medmuskip=2.61mu
\thinmuskip=2.61mu
\thickmuskip=2.61mu
\label{Optimization2}
\argmax_{\alpha_2,\mathcal{P}_2}\log\det{\mathcal{P}_2}~\text{s.t.}~\mathcal{P}_2\succ0,~\alpha_2\geq0,~\text{\eqref{ProbabilisticBoundednessLMI} is satisfied.}
\end{equation}

%%%%%%%%%%%%%%%%%%%%%%%%%%%%%%%%%%%%%%%%%%%%%%%%%%%%%%%%%
\section{Example}
To demonstrate the effectiveness of this analysis for understanding resilience against stealthy attacks, we consider the quadruple tank process, a multivariate laboratory process that consists of four interconnected water tanks \cite{johansson2000quadruple}. This process is an example of a CPS that can be subject to stealthy attacks where an adversary corrupts the inputs and outputs through a memory corruption attack, a man-in-the-middle attack, or an attack on the physical actuators and sensors. The objective of the quadruple tank process is to control the water level of the first two tanks using two pumps. The system has four states (water level for each tank), two inputs (voltages applied to the pumps), and two outputs (voltages from level measurement devices for the first two tanks).

We use the system model in \cite{johansson2000quadruple} at the minimum-phase operating point with a $1.5$ cm$^2$ cross-section for each tank's outlet hole. We discretize the system with a sampling rate of $2$ seconds and use an LQG controller with weights $W=I$ and $V=100I$. To ensure an appropriate noise magnitude, $Q$ and $R$ are created by generating a matrix from a uniform distribution, multiplying it by its transpose, and dividing by $100$. The saturation limits for the actuators are given by $U=\text{BlkDiag}(1/\bar{v}_1^0,1/\bar{v}_2^0)^2$ so that $|\bar{v}_i-\bar{v}_i^0|\leq|\bar{v}_i^0|$, where $\bar{v}_i$ represents the voltage applied to pump $i$ and $\bar{v}_i^0$ represents the operating point of pump $i$. A window size of $T=10$ is used for the $\chi^2$ detector, and the threshold $\eta$ is chosen so that the desired false alarm rate is $1\%$. This system enters an unsafe region of operation when the tanks overflow, so the set of safe states is given by $|h_i-h_i^0|\leq30$ cm, where $h_i$ represents the water level of tank $i$ and $h_i^0$ represents the operating point of tank $i$.

We solve for the maximum noncentrality parameter $\bar{\lambda}$ according to \eqref{NonCentralityMaxSolution} for a $p_d$-stealthy adversary with $p_d=99\%$. Since the optimization problem in \eqref{Optimization1} is not convex, we design the parameters $\gamma$, $\gamma_a$, and $\mathcal{P}_1$ in a suboptimal manner by first finding a $\mathcal{P}_1$ that satisfies $\mathcal{P}_1-\mathcal{A}^T\mathcal{P}_1\mathcal{A}\succeq0$ and then finding the minimal nonnegative value for $\gamma$ such that $\gamma\mathcal{P}_1-\mathcal{A}^T\mathcal{P}_1\mathcal{A}\succeq0$ is satisfied. After that, we solve for $\alpha_1$ and $\gamma_a$ by minimizing $\gamma_a$ subject to \eqref{QuadraticBoundednessLMI}. Since minimizing $\mathcal{E}_2(p)$ jointly over $\alpha_2$ and $\mathcal{P}_2$ in \eqref{Optimization2} is not convex, we find a suboptimal solution for $\alpha_2$ and $\mathcal{P}_2$ by restricting the possible values of $\alpha_2$ to a finite set and maximizing $\max_{\mathcal{P}_2}\log\det{\mathcal{P}_2}$ subject to \eqref{ProbabilisticBoundednessLMI} over that finite set.

Figures \ref{fig:E1Volume1}, \ref{fig:E1Volume2}, and \ref{fig:E1Volume3} depict the volume of $\mathcal{E}_1(k,T_a)$, the ellipsoid in which $\bar{x}_k^1$ lies. As stated in Corollary \ref{VolumeCorollary}, Figure \ref{fig:E1Volume1} shows that the volume of $\mathcal{E}_1(k,T_a)$ is proportional to $T_a$, growing linearly with the total amount of time the system has been under $p_d$-stealthy attack.
\begin{figure}[h!]
\centering
\includegraphics[width=\columnwidth]{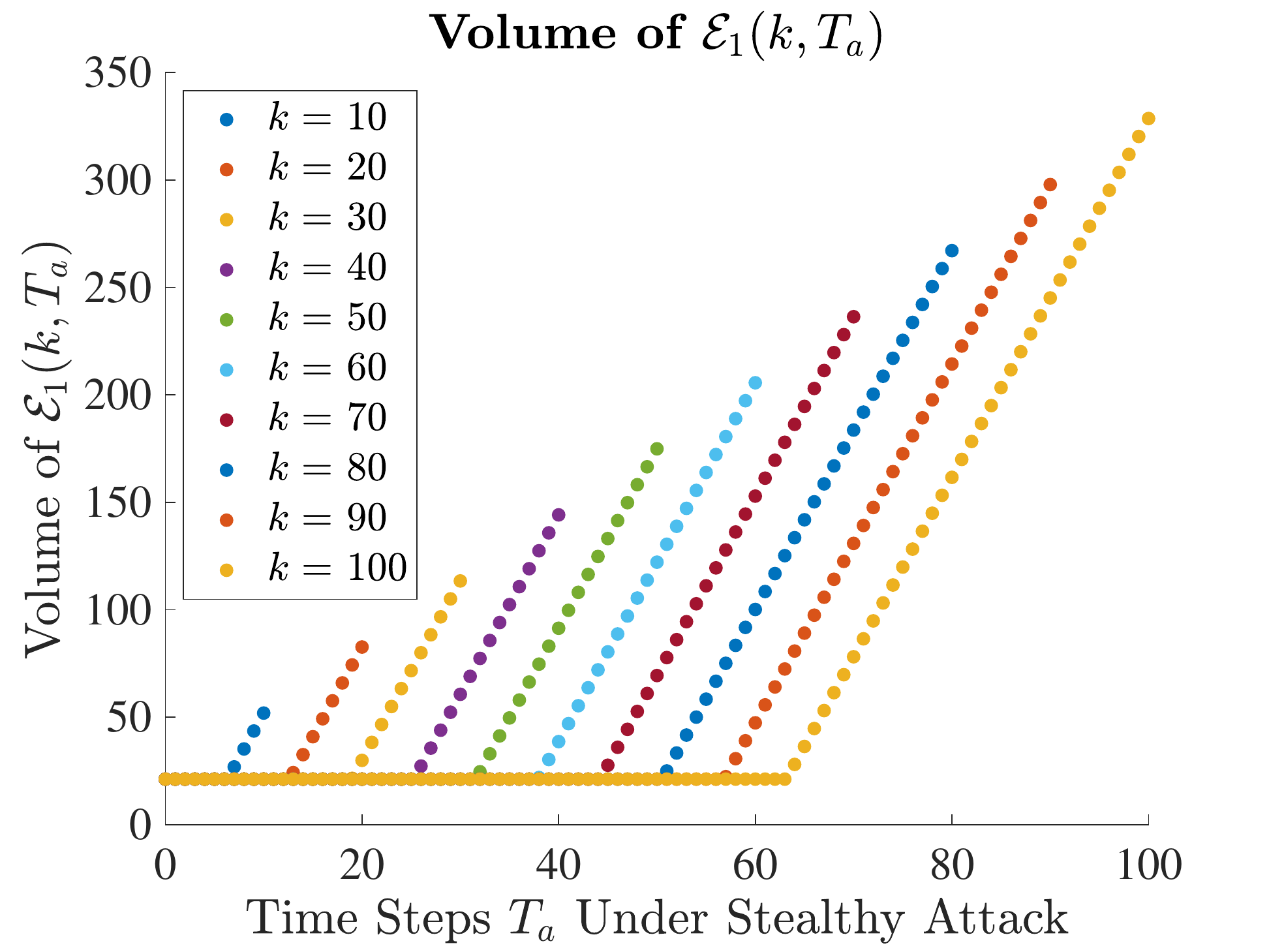}
\caption{Volume of $\mathcal{E}_1(k,T_a)$ as a function of the total number of time steps $T_a$ under $p_d$-stealthy attack.}
\label{fig:E1Volume1}
\end{figure}
Figure \ref{fig:E1Volume2} also confirms the results of Corollary \ref{VolumeCorollary}, showing that the volume of $\mathcal{E}_1(k,T_a)$ is inversely proportional to $k$, decreasing linearly with the total amount of time the system has been under normal operation.
\begin{figure}[h!]
\centering
\includegraphics[width=\columnwidth]{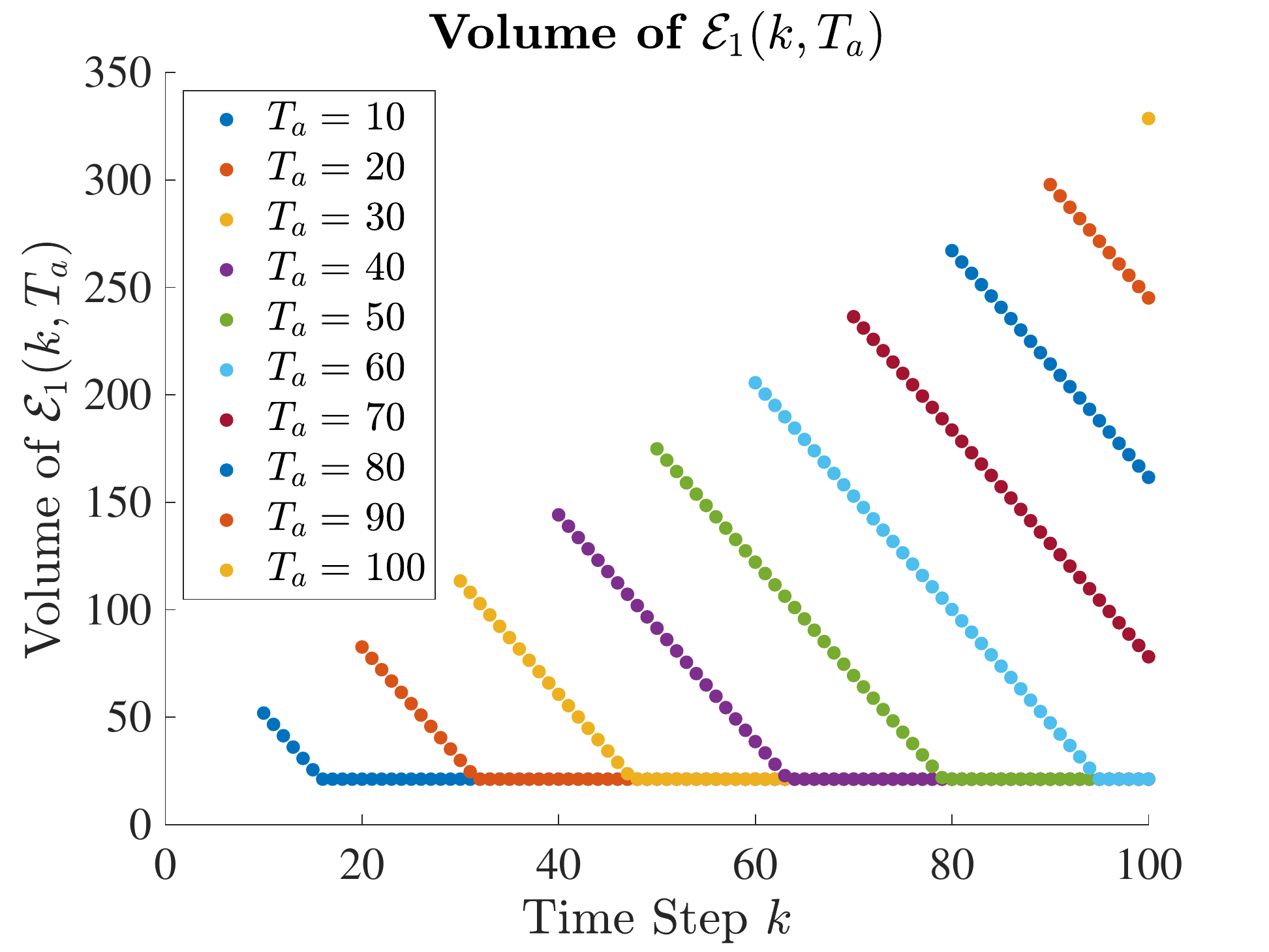}
\caption{Volume of $\mathcal{E}_1(k,T_a)$ as a function of the total number of time steps $k$.}
\label{fig:E1Volume2}
\end{figure}
Figure \ref{fig:E1Volume3} similarly shows how the volume of $\mathcal{E}_1(k,T_a)$ grows linearly with the percentage of time the system is under $p_d$-stealthy attack.
\begin{figure}[h!]
\centering
\includegraphics[width=\columnwidth]{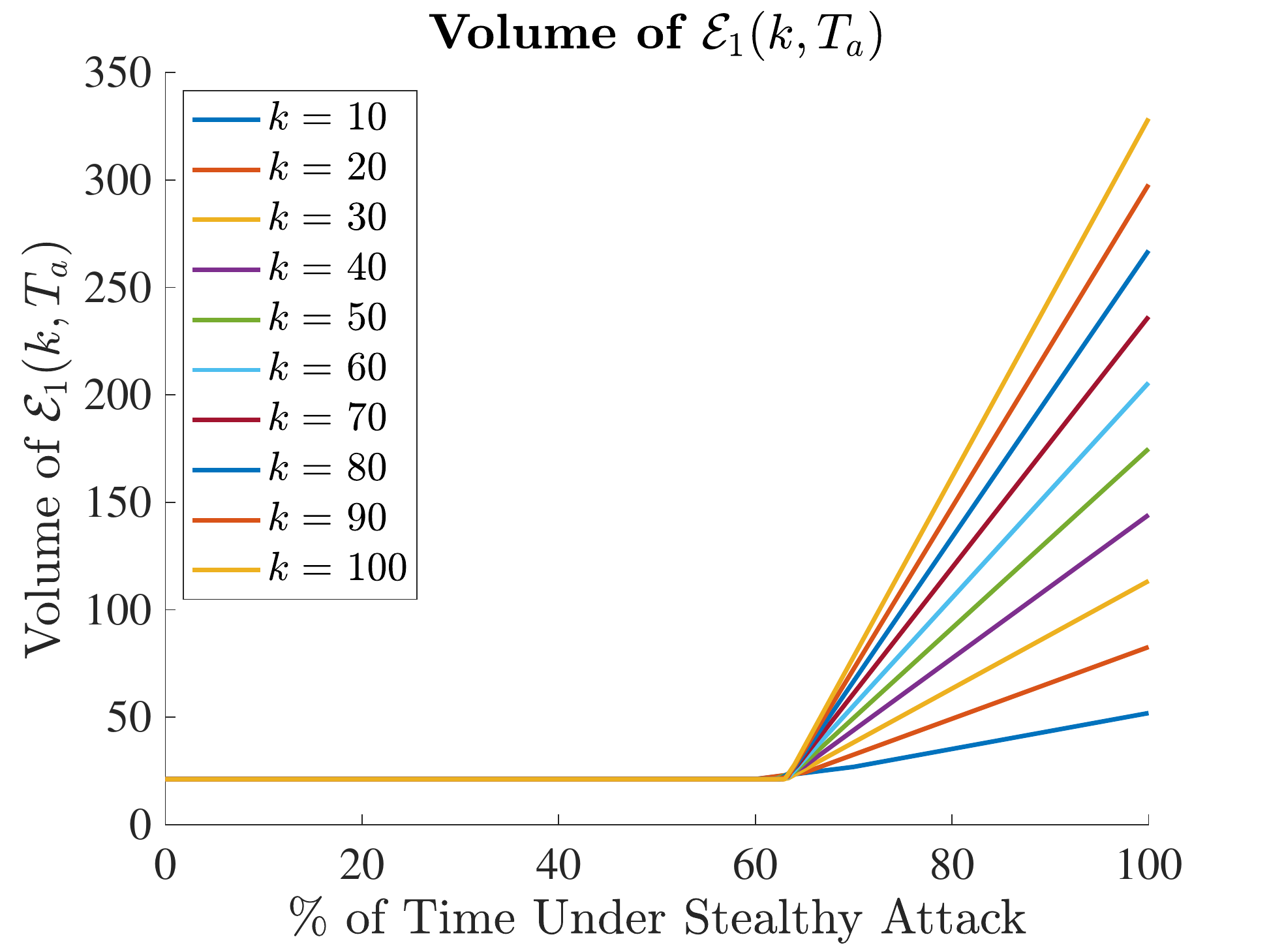}
\caption{Volume of $\mathcal{E}_1(k,T_a)$ as a function of the percentage of time under $p_d$-stealthy attack.}
\label{fig:E1Volume3}
\end{figure}
In Figures \ref{fig:E1Volume1}, \ref{fig:E1Volume2}, and \ref{fig:E1Volume3}, it can be seen that there is a minimal volume for $\mathcal{E}_1(k,T_a)$ which is due to the appearance of the term $\max(1,\gamma^{k-T_a}\gamma_a^{T_a})$ in the definition of $\mathcal{E}_1(k,T_a)$ in \eqref{E1Definition}.

Figure \ref{fig:E2Volume} depicts the volume of $\mathcal{E}_2(p)$, the ellipsoid in which $\bar{x}_k^2$ lies with at least probability $p$.
\begin{figure}[h!]
\centering
\includegraphics[width=\columnwidth]{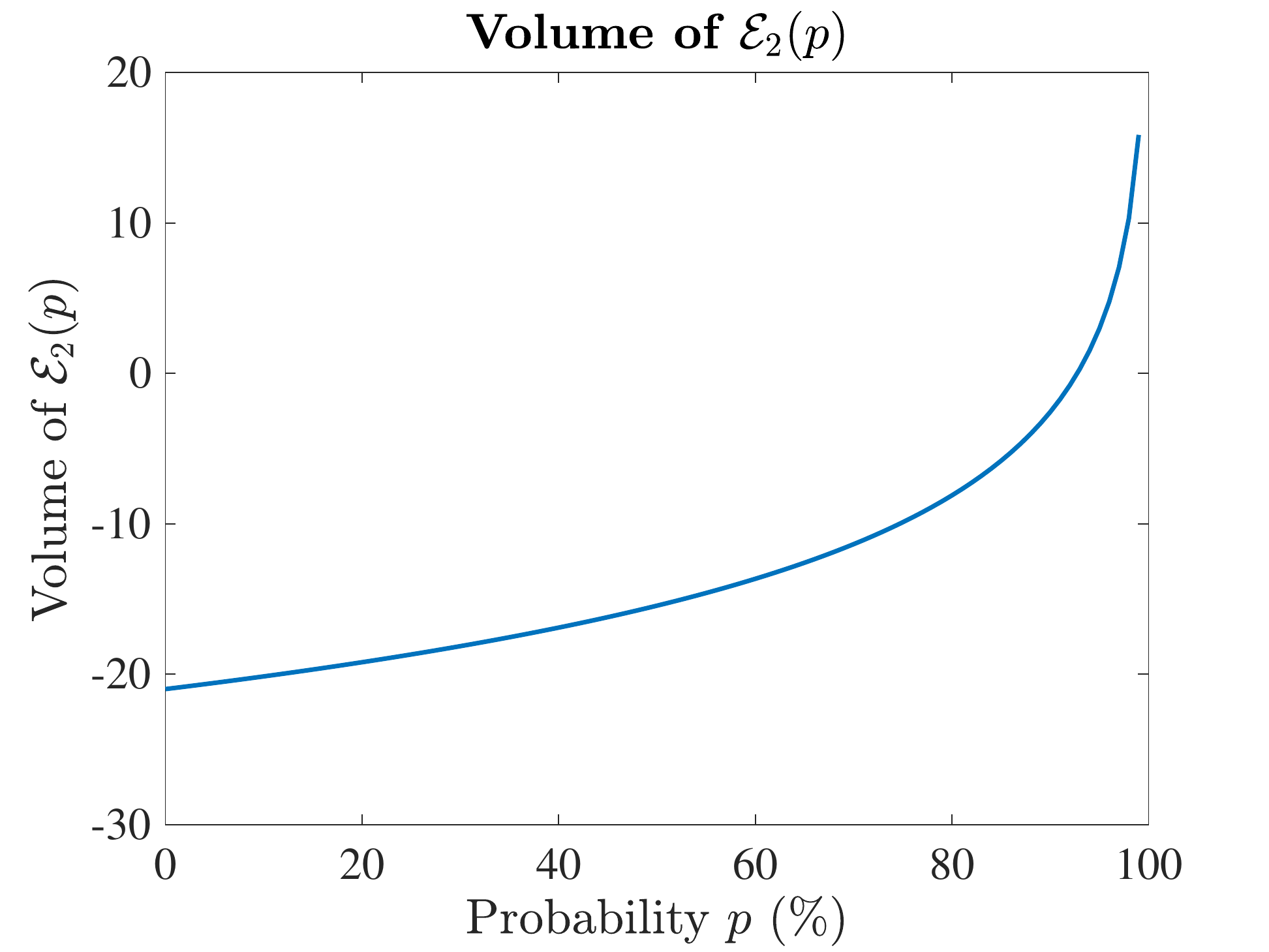}
\caption{Volume of $\mathcal{E}_2(p)$ as a function of the probability level $p$.}
\label{fig:E2Volume}
\end{figure}
As stated in Corollary \ref{VolumeCorollary2}, Figure \ref{fig:E2Volume} shows that the volume of $\mathcal{E}_2(p)$ grows logarithmically with $p$, indicating that larger ellipsoids are associated with greater confidence regions.

Figures \ref{fig:TimingAnalysis1} and \ref{fig:TimingAnalysis2} leverage the safety set presented in Theorem \ref{SafetyTheorem} to indicate whether or not specific response mechanisms will result in maintaining resilience against stealthy adversaries. Figure \ref{fig:TimingAnalysis1} shows the conservative upper bound on the total number of time steps $T_a$ under $p_d$-stealthy attack that the system can remain safe with at least $99\%$ probability.
\begin{figure}[h!]
\centering
\includegraphics[width=\columnwidth]{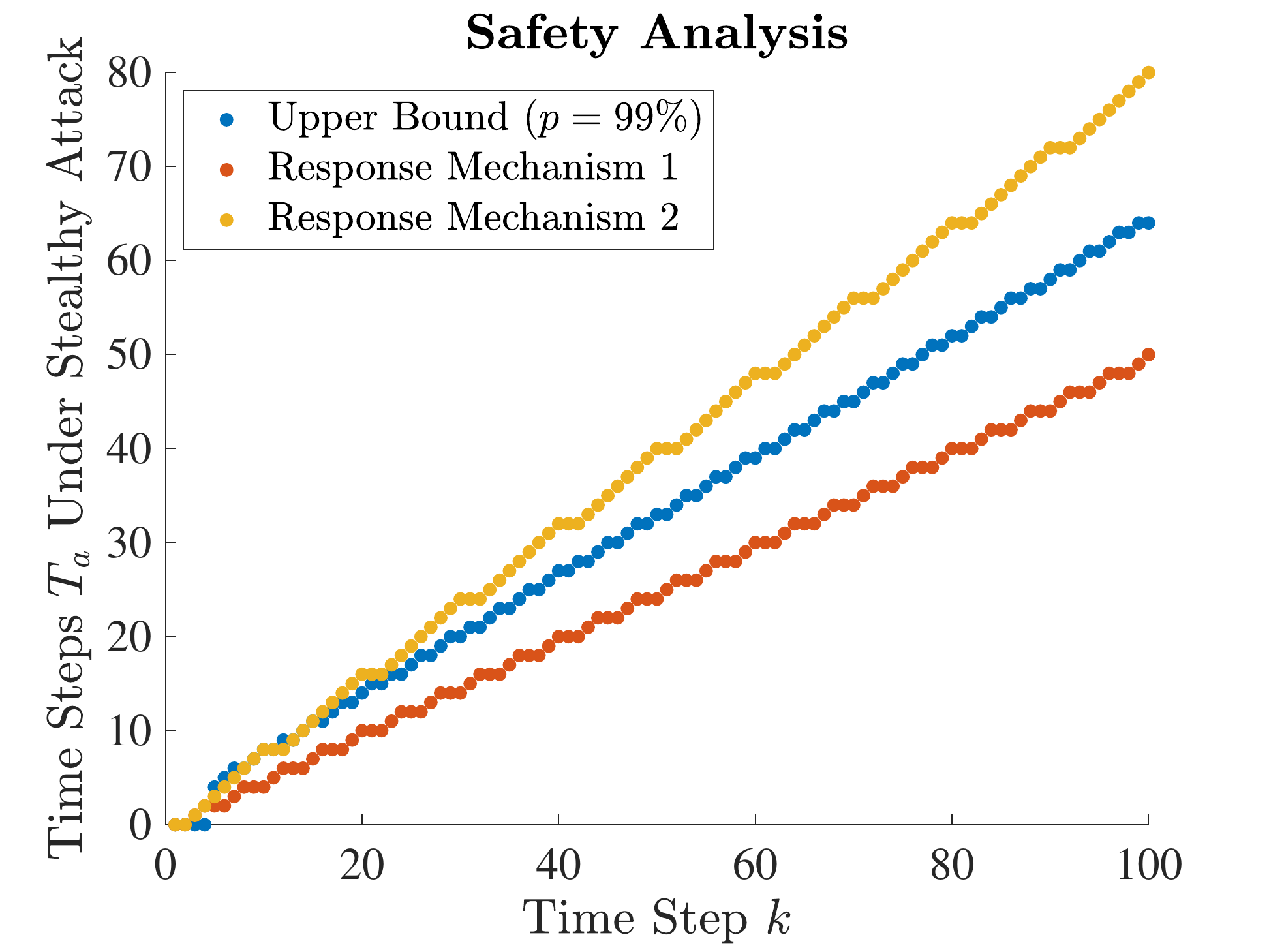}
\caption{Comparison of two different response mechanisms to a conservative upper bound on the total amount of time the system can remain probabilistically safe under $p_d$-stealthy attack in terms of time steps.}
\label{fig:TimingAnalysis1}
\end{figure}
This number of time steps is proportional to the total number of time steps, indicating that the system can handle being under $p_d$-stealthy attack for more time when it has been under normal operation for a larger amount of time. Let $\tau_k^a$ represent the total number of time steps the response mechanism has allowed the system to be vulnerable to attacks at time step $k$. Any response mechanism that ensures $\tau_k^a$ lies below this upper bound will guarantee that the system remains safe against $p_d$-stealthy attacks with at least $99\%$ probability. Figure \ref{fig:TimingAnalysis1} compares $\tau_k^a$ for two response mechanisms to this upper bound. Response mechanism $1$ ensures that every $4$ time steps, the system is not under attack for the first $2$ time steps but then is vulnerable to attacks for the next $2$ time steps. This may be achieved, for example, through software rejuvenation \cite{romagnoli2019design}. Response mechanism $2$ ensures that every $10$ time steps, the system is not under attack for the first $2$ time steps but then is vulnerable to attacks for the next $8$ time steps. As seen in Figure \ref{fig:TimingAnalysis1}, $\tau_k^a$ for response mechanism $1$ lies below the conservative upper bound for safety, guaranteeing that the system will remain safe against $p_d$-stealthy adversaries with at least $99\%$ probability. However, safety cannot be guaranteed when response mechanism 2 is implemented since $\tau_k^a$ lies above the conservative upper bound for safety. The exact same results are shown in Figure \ref{fig:TimingAnalysis2} except that the y-axis depicts the percentage of time under $p_d$-stealthy attack instead of the total number of time steps $T_a$ under $p_d$-stealthy attack.
\begin{figure}[h!]
\centering
\includegraphics[width=\columnwidth]{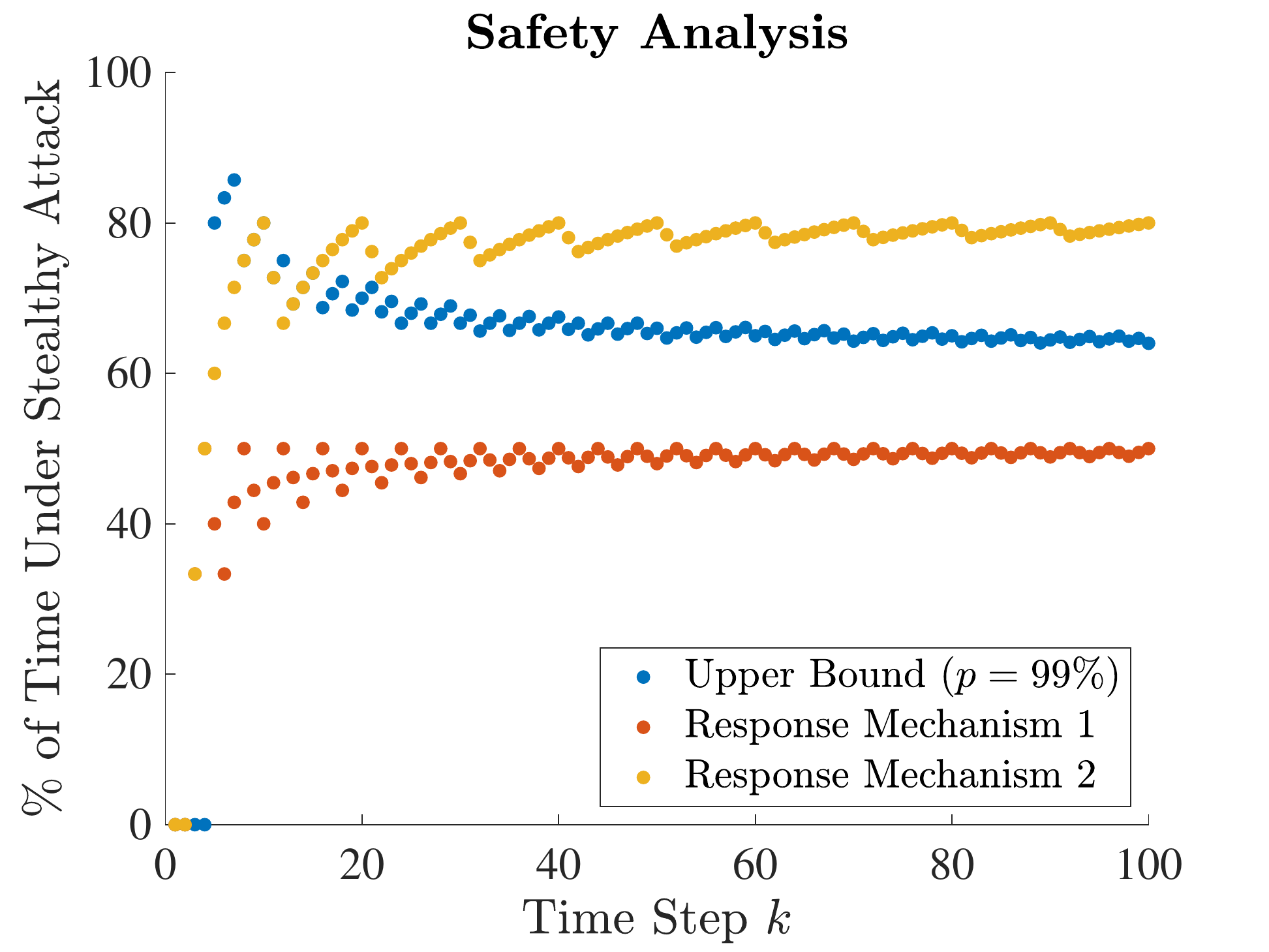}
\caption{Comparison of two different response mechanisms to a conservative upper bound on the total amount of time the system can remain probabilistically safe under $p_d$-stealthy attack in terms of percentages.}
\label{fig:TimingAnalysis2}
\end{figure}

We simulated the quadruple tank process over $1000$ trials for a period of $100$ time steps, allowing the system to be under $p_d$-stealthy attack for the maximum number of time steps according to the conservative upper bound in Figure \ref{fig:TimingAnalysis1}. Over these $1000$ trials, the tanks never overflowed, demonstrating that the conservative upper bound in Figure \ref{fig:TimingAnalysis1} ensures that system remains safe against $p_d$-stealthy adversaries with at least $99\%$ probability.

%%%%%%%%%%%%%%%%%%%%%%%%%%%%%%%%%%%%%%%%%%%%%%%%%%%%%%%%%
\section{Conclusion}
In this article, we have shown that when analyzing attacks on CPSs and designing associated response countermeasures, particular attention must be given to how long the system can remain under attack and still remain safe. We show how a detector limits the set of stealthy biases an adversary can exert on the system, and we use this fact to produce a conservative upper bound on the amount of time the system can be under stealthy attack without violating the safety constraints. This timing analysis can then be used in the design of the response mechanism to thwart adversaries. We illustrate our results with the example of the quadruple tanks process. While this work demonstrates that response mechanisms can be developed to prevent stealthy attacks from violating the safety constraints, future work includes ensuring resilience against attacks that are not stealthy, particularly when there may be missed detections as well as a non-negligible time to detection that causes a delay in initiating the response mechanism. Furthermore, future work includes leveraging online information from the sensor measurements to produce a less conservative upper bound on the amount of time that the system can remain resilient against stealthy attacks.

%%%%%%%%%%%%%%%%%%%%%%%%%%%%%%%%%%%%%%%%%%%%%%%%%%%%%%%%%
\bibliographystyle{IEEEtran}
\bibliography{root}

%%%%%%%%%%%%%%%%%%%%%%%%%%%%%%%%%%%%%%%%%%%%%%%%%%%%%%%%%
\begin{IEEEbiography}[{\includegraphics[width=1in,height=1.25in,clip,keepaspectratio]{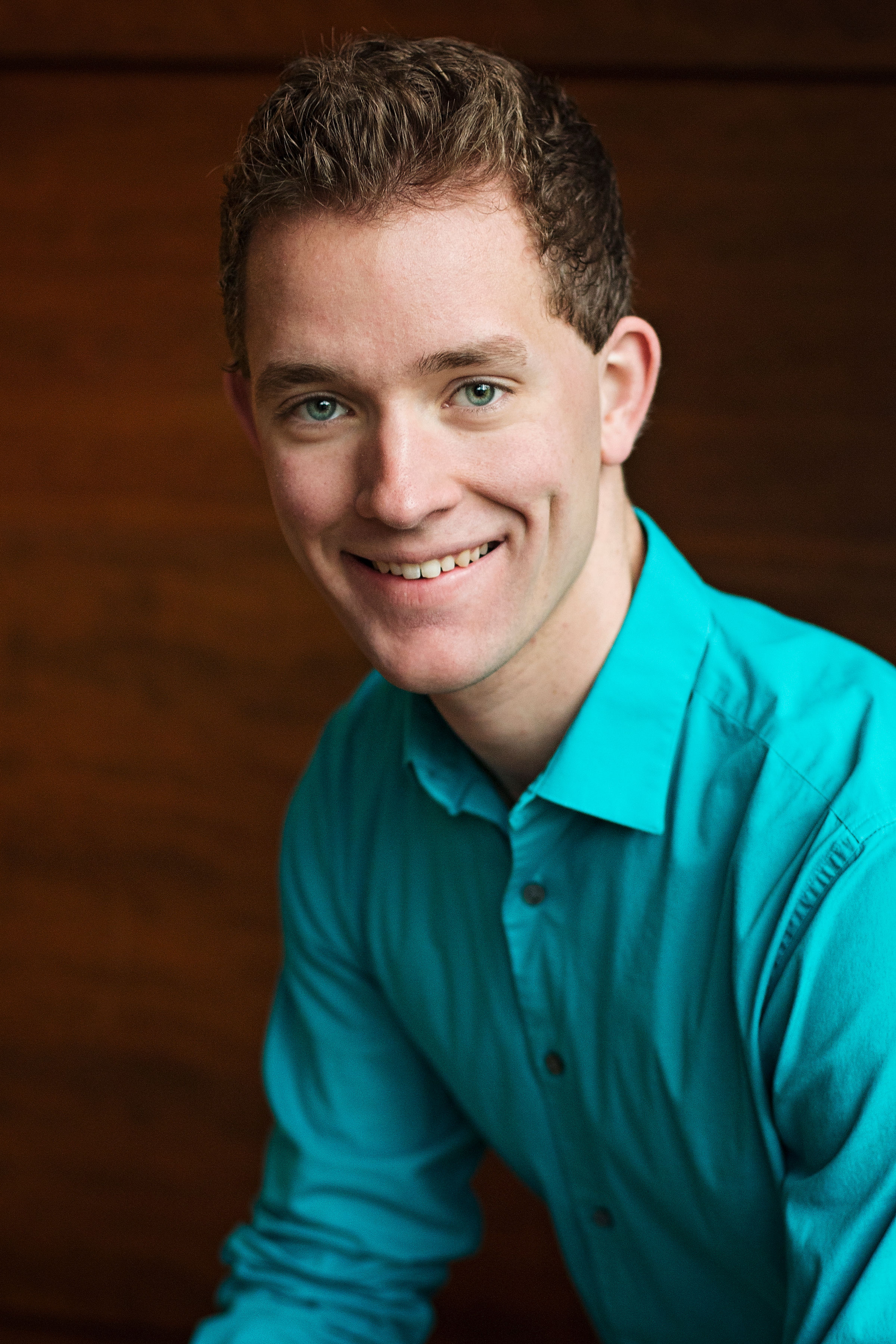}}]{Paul Griffioen}
received the B.S. degree in Engineering, Electrical/Computer concentration, from Calvin College, Grand Rapids, MI, USA in 2016 and the M.S. degree in Electrical and Computer Engineering from Carnegie Mellon University, Pittsburgh, PA, USA in 2018. He is currently pursuing the Ph.D. degree in Electrical and Computer Engineering at Carnegie Mellon University. His research interests include the modeling, analysis, and design of active detection techniques and response mechanisms for ensuring resilient and secure cyber-physical systems.
\end{IEEEbiography}
\begin{IEEEbiography}[{\includegraphics[width=1in,height=1.25in,clip,keepaspectratio]{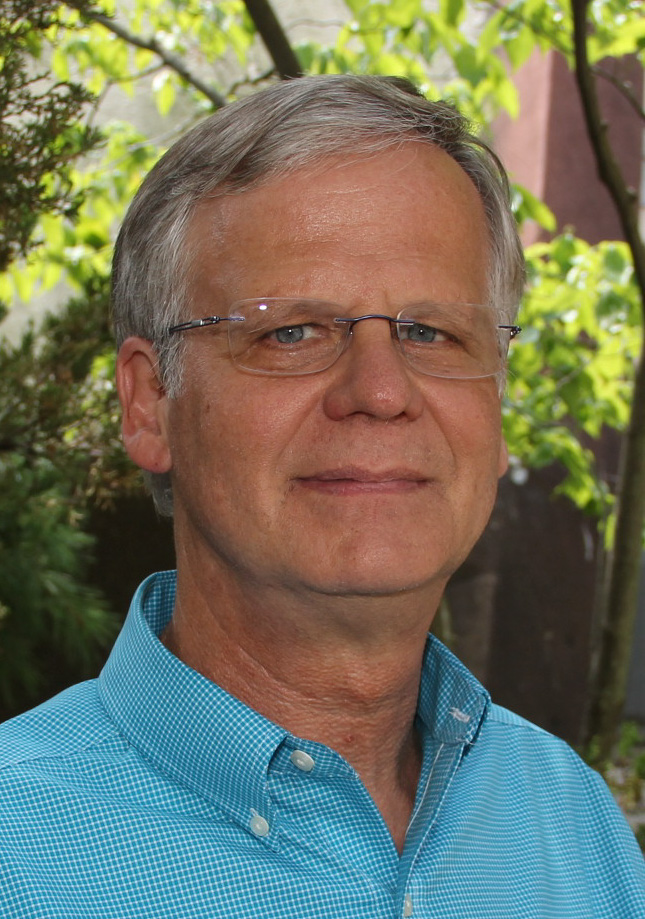}}]{Bruce H. Krogh} is professor emeritus of electrical and computer engineering at Carnegie Mellon University in Pittsburgh, PA, USA, and a member of the technical staff of Carnegie Mellon's Software Research Institute. He was founding director of Carnegie Mellon University-Africa in Kigali, Rwanda. He is chair of the board of the Kigali Collaborative Research Centre (KCRC) in Rwanda and co-lead of the IEEE Continu$\blacktriangleright$ED initiative to develop IEEE's continuing education resources for technical professionals in Africa. Professor Krogh’s research is on the theory and application of control systems, with a current focus on methods for guaranteeing safety and security of cyber-physical systems. He was founding Editor-in-Chief of the \textit{IEEE Transactions on Control Systems Technology}. He is a Life Fellow of the IEEE and a Distinguished Member of the IEEE Control Systems Society.
\end{IEEEbiography}
\begin{IEEEbiography}[{\includegraphics[width=1in,height=1.25in,clip,keepaspectratio]{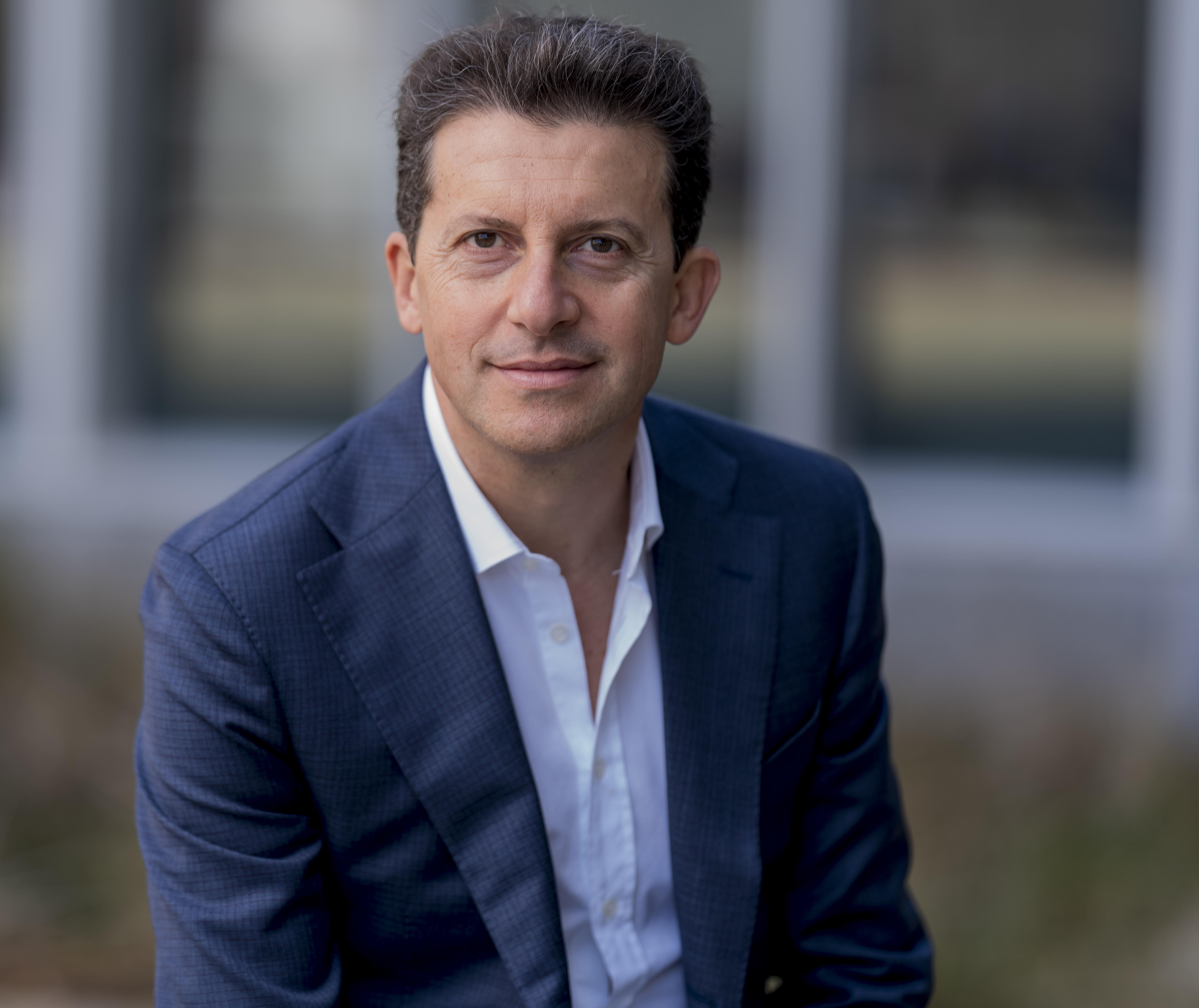}}]{Bruno Sinopoli} is the Das Family Distinguished Professor at Washington University in St. Louis, where he is also the founding director of the center for Trustworthy AI in Cyber-Physical Systems and chair of the Electrical and Systems Engineering Department. He received the Dr. Eng. degree from the University of Padova in 1998 and his M.S. and Ph.D. in Electrical Engineering from the University of California at Berkeley, in 2003 and 2005 respectively. After a postdoctoral position at Stanford University, Dr. Sinopoli was member of the faculty at Carnegie Mellon University from 2007 to 2019, where he was a professor in the Department of Electrical and Computer Engineering with courtesy appointments in Mechanical Engineering and in the Robotics Institute and co-director of the Smart Infrastructure Institute. His research interests include modeling, analysis and design of Resilient Cyber-Physical Systems with applications to Smart Interdependent Infrastructures Systems, such as Energy and Transportation, Internet of Things and control of computing systems.
\end{IEEEbiography}

\end{document}